\documentclass[10pt,english,oneside,a4paper]{amsart}

\usepackage{concmath}
\usepackage{tikz-cd}

\usepackage[utf8x]{inputenc}
\usepackage[T1]{fontenc}
\usepackage{varioref}
\usepackage{amsmath}
\usepackage{amssymb}

\makeatletter
\theoremstyle{plain}
\theoremstyle{plain}
\newtheorem{theorem}{Theorem}
  \theoremstyle{plain}
  \newtheorem{lemma}[theorem]{Lemma}
  \theoremstyle{plain}
  \newtheorem{definition}[theorem]{Definition}
  \theoremstyle{plain}
  \newtheorem{proposition}[theorem]{Proposition}
  \theoremstyle{remark}
  \newtheorem*{note*}{Note}
  \theoremstyle{remark}
  \newtheorem*{conclusion*}{Conclusion}
  \theoremstyle{remark}
  \newtheorem{remark}[theorem]{Remark}
  \theoremstyle{definition}

  \theoremstyle{plain}
  \newtheorem{corollary}[theorem]{Corollary}

\usepackage{amsfonts}
\usepackage{amscd}
\usepackage[cmtip,arrow]{xy}
\usepackage{pb-diagram,pb-xy}
\usepackage{hyperref}

\usepackage[shortlabels]{enumitem}

\newcommand{\kf}{\mathfrak{k}}

\newcommand{\pf}{\mathfrak{p}}

\newcommand{\cI}{{\mathcal I}}
\newcommand{\cJ}{{\mathcal J}}
\newcommand{\cL}{{\mathcal L}}

\newcommand{\mR}{\mathbb{R}}
\newcommand{\mC}{\mathbb{C}}

\makeatother

\usepackage{babel}

\begin{document}

\title{Unified formalism for Palatini gravity}

%%% For REvTex4 template
%\author{S. Capriotti}
%\affiliation{Departamento de Matem\'atica, \mbox{Universidad Nacional del Sur}, 8000 Bah\'{\i}a Blanca, Argentina.}
%%% For Article template
\author{S. Capriotti}
\address{Departamento de Matem\'atica \\
Universidad Nacional del Sur, CONICET \\
 8000 Bah{\'\i}a Blanca \\
Argentina}

\begin{abstract}
  The present article is devoted to the construction of a unified formalism for Palatini and unimodular gravity. The basic idea is to employ a relationship between unified formalism for a Griffiths variational problem and its classical Lepage-equivalent variational problem. As a way to understand from an intuitive viewpoint the Griffiths variational problem approach considered here, we may say the variations of the Palatini Lagrangian are performed in such a way that the so called \emph{metricity condition}, i.e. (part of) the condition ensuring that the connection is the Levi-Civita connection for the metric specified by the vielbein, is preserved. From the same perspective, the classical Lepage-equivalent problem is a geometrical implementation of the Lagrange multipliers trick, so that the metricity condition is incorporated directly into the Palatini Lagrangian. The main geometrical tools involved in these constructions are canonical forms living on the first jet of the frame bundle for the spacetime manifold $M$. These forms play an essential r\^ole in providing a global version of the Palatini Lagrangian and expressing the metricity condition in an invariant form; with their help, it was possible to formulate an unified formalism for Palatini gravity in a geometrical fashion. Moreover, we were also able to find the associated equations of motion in invariant terms and, by using previous results from the literature, to prove their involutivity. As a bonus, we showed how this construction can be used to provide a unified formalism for the so called \emph{unimodular gravity} by employing a reduction of the structure group of the principal bundle $LM$ to the special linear group $SL\left(m\right),m=\mathop{\text{dim}}{M}$.
\end{abstract}



\email{santiago.capriotti@uns.edu.ar}

\thanks{
  The author thanks CONICET for finantial support, and as a member of research projects PIP 11220090101018 and PICT 2010-2746 of the ANPCyT. 
}

\maketitle

\setcounter{tocdepth}{1}

\tableofcontents

\section{Introduction}
The search of a Hamiltonian setting for General Relativity has a long and distinguished history. Among the first works dealing with this problem, we can mention \cite{PhysRev.79.986,PhysRev.87.452,PhysRev.83.1018} and the works from Dirac \cite{10.2307/100496,Dirac1958}. A coordinate-free formulation appeared in 1962 with the groundbreaking work of Arnowitt, Deser and Misner (a reprinting of this article can be found in \cite{citeulike:820116}). For a modern account, see for example \cite{bojowald2010canonical,9780521830911,Thiemann:2007zz} and references therein. However, as these pictures require of a $3+1$-decomposition of space-time, they tend to hide the full covariance of the theory. 

From a general viewpoint, the so called \emph{multisymplectic description of field theory} was developed as a way to preserve this covariance (see \cite{dedecker1953calcul,Kijowski1973,kijowski79:_sympl_framew_field_theor,Carinena1991345,Gotay:2004ib,helein04:_hamil,9789810215873,helein:hal-00599691} and references therein). The formulation of general relativity from this geometrical viewpoint, both in the Lagrangian and the Hamiltonian realm, was carried out in several places (we can mention, for example, \cite{0264-9381-22-19-016,Vignolo2006,doi:10.1063/1.4890555,Ibort:2016xoo}); a purely multisymplectic formulation was given in \cite{0264-9381-32-9-095005}.

Nevertheless, the singular nature of the Lagrangians associated to field theory requires special techniques in order to succesfully construct the Hamiltonian counterpart of a variational problem. A way to overcome these problems is by means of the so called \emph{Skinner-Rusk or unified formalism} \cite{DeLeon:2002fc,2004JMP....45..360E,2014arXiv1402.4087P,375178}. Relevant features of this formulation are:
\begin{itemize}
\item In it, both the velocities and the momenta are present as degrees of freedom, and
\item in the underlying variational problem, the variations of the velocities are performed independently of the variations of the fields.
\end{itemize}

The present work continues the exploration of geometrical formulations of the Lagrangian and Hamiltonian versions of General Relativity (GR from now on). More specifically, it deals with a unified formalism \cite{1751-8121-40-40-005,2004JMP....45..360E,Vitagliano2010857} for Palatini version of GR, pursuing the work initiated in \cite{Gaset:2017ahy} for Einstein-Hilbert action. Also, a novel unified version of \emph{unimodular gravity} is obtained as a by-product of the geometrical tools employed in the article.

Our starting point will be the Griffiths variational problem considered in \cite{capriotti14:_differ_palat} for Palatini gravity; this variational problem can be considered as an alternative Lagrangian version for the formulation of vacuum GR equations in terms of exterior differential systems, as it is given for example in \cite{PhysRevD.71.044004}. The method chosen for the construction of the unified formalism comes from the work of Gotay \cite{GotayCartan,Gotay1991203}, and it consists into the employment of a \emph{classical Lepage-equivalent variational problem} (in the sense defined by Gotay in these works) as a replacement for the Griffiths variational problem at hands. The new problem becomes a particular instance of the $m$-phase space theory in the sense of Kijowski \cite{Kijowski1973}.

Roughly speaking, a Griffiths variational problem is a variational problem in which the variations are selected to preserve a specific relationship between the involved fields; in the particular case of Palatini gravity, the degrees of freedom are a moving frame and a connection, and they can be interpreted as forming an element in the jet bundle of the frame bundle (see Equation \eqref{eq:JetToConnectionPlusFrame} and Theorem \ref{thm:metr-constr-interpretation} below). In this setting our variations will be performed in such a way that the connection is always the Levi-Civita connection for the metric associated to the vielbein. It will be achieved here by assuming the \emph{metricity condition} as the basic constraint; this condition is equivalent to the annihilation of the so called \emph{nonmetricity tensor}, as it appeared in \cite{Friedric-1978} (for a more explicit explanation on this, see Section \ref{sec:null-trace-constr}). It should be stressed that this condition is weaker than the most common restriction found when working with variational problems on a jet space, namely, the set of constraints imposed on the fields by the \emph{contact structure} (see Section \ref{sec:geom-prel}). On the other hand, the classical Lepage-equivalent is a geometrical construction equivalent to the well-known ``Lagrange multiplier trick'', where the constraints on the fields are incorporated as terms into the Lagrangian of the theory. For articles employing this trick, see \cite{Ray1975,Safko1976,Friedric-1978}, although in these references the components of the metric tensor instead of vielbeins are used as degrees of freedom besides connection variables. Through it, we will be able to find a unified formulation for Palatini and unimodular GR.

In more geometrical terms, the definition of \emph{classical Lepage-equivalent problem} goes as follows: We begin with a bundle $p:F\rightarrow M$ with $m:=\dim{M}$, an $m$-form $\lambda$ on $F$ and a set of restrictions encoded as an exterior differential system $\cI$ on $F$ \cite{nkamran2000,BryantNine,CartanBeginners,BCG}; furthermore, it is necessary to admit that $\cI$ is locally generated by sections of a subbundle $I\subset\wedge^\bullet\left(T^*F\right)$. Under these assumptions, Gotay showed in the previously cited works how to construct another variational problem, its \emph{classical Lepage-equivalent problem}, whose underlying bundle is the affine subbundle of forms
\[
W_\lambda:=\left(\lambda+I\right)\cap\wedge^m\left(T^*F\right),
\]
and where the new Lagrangian form is calculated as the pullback of the canonical $m$-form on $\wedge^m\left(T^*F\right)$ to $W_\lambda$. The degrees of freedom along the fibers of $I$ play the r\^ole of Lagrange multipliers, and the Lagrangian of the theory is recovered through the pullback construction described above. It is possible to show that in this new variational problem, any extremal of the Lepage-equivalent problem projects onto extremal sections of the original variational problem; in the language used in \cite{GotayCartan}, it is said that the classical Lepage-equivalent problem is \emph{covariant}\footnote{We cannot avoid here the accidental lack of uniqueness in the terminology; the covariance property of a Lepage-equivalent problem does not have direct relationship with the invariance of the underlying variational problems respect to general changes of coordinates.}. However, there is no general proof of the \emph{contravariance} of a Lepage-equivalent variational problem, namely, the fact that every critical section of the original variational problem can be lifted to a critical section of the Lepage-equivalent, and so it must be proved in each case separately.

A key feature of the previous scheme is that, when $F=J^1\pi$ for some bundle $\pi:E\rightarrow M$, $\lambda$ is a Lagrangian density and the exterior differential system $\cI$ is the contact structure in $J^1\pi$, the classical Lepage-equivalent problem yield to the unified formalism as it is described in \cite{2004JMP....45..360E}.

Thus, for setting a suitable unified formalism for Palatini gravity, we will apply the Gotay scheme; this brings us to the problem of constructing a classical Lepage-equivalent problem for the Griffiths variational problem considered in \cite{capriotti14:_differ_palat}, and to show its contravariance. Part of the present article is devoted to this task. In the final part we will show how to obtain the equations of motion from the formalism, and how the existence of a fundamental volume form gives rise to a unified formalism for unimodular gravity.

The organization of the paper is the following: In Section \ref{sec:geom-prel} the geometrical tools and conventions to be employed throughout the article are introduced, as well as the relevant definitions regarding Griffiths variational problems and its classical Lepage-equivalents. The canonical forms on the jet space of the bundle of frames are also introduced in this section. The unified formalism for Palatini gravity is described in Section \ref{sec:griff-vari-probl}, where a discussion about the metricity constraint is carried out. To work with the equations of motion of the unified formalism requires the choice of a basis for the vertical vector fields of the underlying bundle. To this end is devoted Section \ref{sec:convenient-basis}: The selection of a connection on $LM$ allows us to construct a basis on $LM$, and its elements are lifted to the jet space through the canonical lifts. In the first part of Section \ref{sec:veloc-mult-space} the existence of a direct product structure on the bundle of forms belonging to the basic bundle of the unified formalism is used to find a basis of vector fields on this bundle, suitable to work with the equations of motion. The equations of motion are explicity written in the final part of this section, and its involutivity analised. Finally, an unified formalism for unimodular gravity is discussed in Section \ref{sec:mult-unim-grav}.

\section{Geometrical preliminaries}
\label{sec:geom-prel}

The spacetime manifold will be indicated with $M$, and it will have dimension $m$. Throughout the article, lower case latin indices $i,j,k,l,m,\cdots$ will refer to coordinates in the tensor products of the vector space $\mR^m$ and its dual; they will run from $1$ to $m$. With this convention in mind, the canonical basis in $\mR^m$ will be indicated as $\left\{e_i\right\}$, and its dual with $\left\{e^j\right\}$. Greek indices, on the other hand, will refer to indices associated to local coordinates in the spacetime manifold $M$; for this reason, they also will run from $1$ to $m$. Finally, upper case latin indices will be used for the representation of general fiber coordinates. Einstein convention regarding sum over repeated indices will be adopted.

The matrix
\[
\eta:=
\begin{bmatrix}
  1&\cdots&0&0\\
  \vdots&\ddots&\vdots&\vdots\\
  0&\cdots&1&0\\
  0&\cdots&0&-1
\end{bmatrix}\in GL\left(m\right)
\]
will set the metric on $\mR^m$. It should be stressed that there is nothing special in the signature chosen for $\eta$, and that the results reached in the article will work for any other signature.

The matrix $\eta$ determines a Lie algebra
\[
\mathfrak{u}\left(m-1,1\right):=\left\{A\in\mathfrak{gl}\left(m,\mC\right):A^\dagger\eta+\eta A=0\right\},
\]
which is a compact real form for the complexification $\mathfrak{gl}\left(m,\mC\right)=\mathfrak{gl}\left(m\right)\otimes_\mR\mC$. Another way to define this compact form is through the involution
\[
F:\mathfrak{gl}\left(m,\mC\right)\rightarrow\mathfrak{gl}\left(m,\mC\right):A\mapsto-\eta A^\dagger\eta;
\]
the eigenspaces of $F$, associated to the eigenvalues $\pm1$, induce the decomposition 
\[
\mathfrak{gl}\left(m,\mC\right)=\mathfrak{u}\left(m-1,1\right)\oplus\mathfrak{s}\left(m-1,1\right).
\]
This decomposition is the \emph{Cartan decomposition of $\mathfrak{gl}\left(m,\mC\right)$ associated to the compact real form $\mathfrak{u}\left(m-1,1\right)$}, and descends to $\mathfrak{gl}\left(m\right)\subset\mathfrak{gl}\left(m,\mC\right)$, namely
\[
\mathfrak{gl}\left(m\right)=\kf\oplus\pf,
\]
where
\[
\kf:=\mathfrak{u}\left(m-1,1\right)\cap\mathfrak{gl}\left(m\right),\qquad\pf:=\mathfrak{s}\left(m-1,1\right)\cap\mathfrak{gl}\left(m\right).
\]
Denoting $f:=F_{|\mathfrak{gl}\left(m\right)}$, we have that $\kf$ (resp. $\pf$) is the eigenspace corresponding to the eigenvalue $+1$ (resp. $-1$) for $f$. The projectors in every of these eigenspaces become
\[
\pi_\kf\left(A\right):=\frac{1}{2}\left(A-\eta A^T\eta\right),\qquad\pi_\pf\left(A\right):=\frac{1}{2}\left(A+\eta A^T\eta\right).
\]

There exists some facts related to this decomposition which could be useful when dealing with $\mathfrak{gl}\left(m\right)$-valued forms. First of all, given  $N$ a manifold and $\gamma\in\Omega^p\left(N,\mathfrak{gl}\left(m\right)\right)$, we will define
\[
\gamma_\kf:=\pi_\kf\circ\gamma,\qquad\gamma_\pf:=\pi_\pf\circ\gamma;
\]
if $\gamma=\gamma^i_jE^j_i$ is the expression of $\gamma$ in terms of the canonical basis of $\mathfrak{gl}\left(m\right)$, then we will have that
\begin{align*}
  &\left(\gamma_\kf\right)^i_j=\frac{1}{2}\left(\gamma^i_j-\eta_{jp}\gamma^p_q\eta^{qi}\right)\\
  &\left(\gamma_\pf\right)^i_j=\frac{1}{2}\left(\gamma^i_j+\eta_{jp}\gamma^p_q\eta^{qi}\right).
\end{align*}
Additional properties for this decomposition can be found in Appendix \ref{sec:cart-decomp-forms}.

We will make extensive use of jet bundle theory throughout the article, as it is presented in \cite{saunders89:_geomet_jet_bundl}. Thus, associated to every bundle $\pi:E\rightarrow M$ there exists an affine bundle $J^1\pi$ and maps $\pi_1:J^1\pi\rightarrow M,\pi_{10}:J^1\pi\rightarrow E$ fitting in the following commutative diagram
\begin{center}
  \begin{tikzcd}[row sep=1.2cm]
    J^1\pi
    \arrow{rr}{\pi_{10}}
    \arrow[swap]{dr}{\pi_1}
    &
    &
    E
    \arrow{dl}{\pi}
    \\
    &
    M
    &
  \end{tikzcd}
\end{center}
The elements of $J^1\pi$ are regarded as linear maps
\[
j_x^1s:T_xM\rightarrow T_eE
\]
such that $T_e\pi\circ j_x^1s=\text{id}_{T_xM}$. Every section $s:M\rightarrow E$ can be lifted to a section $j^1s:M\rightarrow J^1\pi$ through the formula
\[
j^1s\left(x\right):=T_xs.
\]
Sections of $\pi_1$ arising as lifts of section of $\pi$ are called \emph{holonomic sections}.

The \emph{contact structure on $J^1\pi$} is the ideal in $\Omega^\bullet\left(J^1\pi\right)$ locally generated by the contact forms
\[
du^A-u^A_\mu dx^\mu
\]
and its differentials. According to the next result, it fully characterizes the holonomic sections of $\pi_1$.
\begin{proposition}\label{prop:HolonomicSections}
  A section $s$ of $\pi_1$ is holonomic if and only if
  \[
  \sigma^*\left(du^A-u^A_\mu dx^\mu\right)=0.
  \]
\end{proposition}

In the construction of the classical Lepage-equivalent problem associated to a given Griffiths variational problem it will be also necessary to have at our disposal some facts regarding spaces of forms on a fiber bundle $\pi:E\rightarrow M$. For every $k<l$, the set of $k$-horizontal $l$-forms on $E$ is defined by
\[
\left.\wedge^l_k\left(T^*E\right)\right|_e:=\left\{\alpha\in\wedge^l\left(T^*_eE\right):V_1\lrcorner\cdots\lrcorner V_k \lrcorner\alpha=0\quad\forall V_1,\cdots,V_k\in V_e\pi\right\}.
\]
We will indicate with
\[
i_k^l:\wedge^l_k\left(T^*E\right)\hookrightarrow\wedge^l\left(T^*E\right)
\]
the canonical immersions. The canonical $m$-form $\Theta$ on $\wedge^m\left(T^*E\right)$ is defined by the expression
\[
\left.\Theta\right|_{\alpha_e}\left(Z_1,\cdots,Z_m\right):=\alpha_e\left(T_{\alpha_e}\overline{\tau}\left(Z_1\right),\cdots,T_{\alpha_e}\overline{\tau}\left(Z_m\right)\right)
\]
Further properties of this form can be found in Appendix \ref{sec:canonical-k-form}. We will also set
\[
\Theta_2:=\left(i_2^m\right)^*\Theta,\qquad\Omega_2:=-d\Theta_2.
\]

In order to find local expressions let us choose adapted local coordinates $\left(x^\mu,u^A\right)$ on $U\subseteq E$. Then we have local coordinates $\left(x^\mu,u^A,p,p_A^\mu\right)$ on $\left(\overline{\tau}^m_E\right)^{-1}\left(U\right) \subseteq\wedge^m_2\left(T^*E\right)$ given by the condition $\gamma\in\left(\overline{\tau}^m_E\right)^{-1}\left(U\right)$ if and 
only if
\[
\gamma=pd^mx+p_A^\mu du^A\wedge d^{m-1}x_\mu
\]
where
\[
d^mx:=dx^1\wedge\cdots\wedge dx^m,\qquad d^{m-1}x_\mu:=\partial_\mu\lrcorner d^mx,\qquad\partial_\mu\equiv\frac{\partial}{\partial x^\mu}.
\]
In terms of these coordinates one has
\begin{align*}
  \Theta_2 \left(x^\mu,u^A,p,p_A^\mu\right) &=pd^mx+p_A^\mu du^A\wedge d^{m-1}x_\mu,\\
  \Omega_2 \left(x^\mu,u^A,p,p_A^\mu\right)&=-dp\wedge d^mx-dp_A^\mu\wedge du^A\wedge d^{m-1}x_\mu.
\end{align*}

\subsection{Griffiths variational problems}
\label{sec:teoria-lagrangiana}

These kind of variational problems were considered by Griffiths in \cite{book:852048}, and have been employed in geometry \cite{hsu92:_calcul_variat_griff,sabau_shibuya_2016,10.2307/2374654} and mathematical physics \cite{0264-9381-24-22-005,makhmali16:_differ}. The essential data for the construction of this version of variational theory is the following:

\begin{itemize}
\item A fiber bundle $p:F\rightarrow M$.
\item An $m$-form $\lambda\in\Omega^m\left(F\right)$, the \emph{Lagrangian form}.
\item An exterior differential system (EDS from now on) $\cI\in\Omega^\bullet\left(F\right)$, that is, an ideal in the exterior algebra of $F$ that is closed by exterior differentiation. 
\end{itemize}

\begin{remark}
  For first order Lagrangian field theory, the bundle $F$ is set to be the jet bundle $J^1\pi$ associated to a bundle $\pi:E\rightarrow M$, the Lagrangian form is in general a horizontal form $\lambda=L\pi_1^*\nu$, where $L\in C^\infty\left(J^1\pi\right)$ and $\nu$ is a volume form on $M$. In this case, the EDS becomes the contact structure of $J^1\pi$.
\end{remark}

\begin{definition}
  The \emph{Griffiths variational problem associated to the data $\left(F,\lambda,\cI\right)$} consists into finding a section
  $\sigma:M\rightarrow F$ stationary for the action
  \[
  S\left[\sigma\right]:=\int_M\sigma^*\lambda
  \]
  which is an integral section of $\cI$, namely, such that
  \[
  \sigma^*\alpha=0
  \]
  for every $\alpha\in\cI$. A section fulfilling these requeriments for a given variational problem is also called \emph{critical}.
\end{definition}

\begin{remark}
  We will assume the existence of these integrals as granted.
\end{remark}

\begin{remark}
  For first order Lagrangian field theory, Proposition \ref{prop:HolonomicSections} implies that a section of $J^1\pi$ will be integral for the contact structure if and only if it is holonomic. Thus the associated variational problem will translate into finding the sections $s:M\rightarrow E$ such that their lifts $j^1s:M\rightarrow J^1\pi$ are stationary for the action integral
  \[
  S\left[s\right]=\int_M\left(j^1s\right)^*\left(L\pi^*\nu\right)=\int_ML\left(j^1s\right)\nu,
  \]
  which is the Hamilton's principle \cite{Gotay:2004ib,campos10:_geomet_method_class_field_theor_contin_media} for this kind of field theories.
\end{remark}

\subsection{Unified formalism for a Griffiths variational problem}
\label{sec:form-hamilt-restr}

In order to construct a unified formalism for a given Griffiths variational problem $\left(F,\lambda,\cI\right)$, we need to assume that the EDS $\cI$ is locally generated by the set of sections of a vector subbundle $I\subset\wedge^\bullet\left(T^*F\right)$; it means that we can find an open cover $\left\{U_\alpha\right\}$ of $F$ such that the pullback of $\cI$ to each of the elements $U_\alpha$ of the cover is generated by sections of $I_{|U_\alpha}$. The following definitions are quoted from \cite{GotayCartan}.

Let us fix some integer $k$ such that
\[
\left.\lambda\right|_p\in\wedge^m_k\left(T_pF\right),\qquad I\subset\wedge^m_k\left(T^*F\right).
\]
We define the affine subbundle $W_\lambda\subset\wedge^m\left(T^*F\right)$ with the formula
\[
\left.W_\lambda\right|_p:=\left.\lambda\right|_p+I^m_p,
\]
where $I^m_p:=I\cap\wedge_k^m\left(T^*_pF\right)$ is the fiber composed of the $k$-horizontal $m$-forms of $I$ at $p\in F$. Also, we define the $m$-form $\Theta_\lambda$ as the pullback of the canonical $m$-form $\Theta\in\Omega^m\left(\wedge^m\left(T^*F\right)\right)$ to $W_\lambda$, and
\[
\Omega_\lambda:=d\Theta_\lambda.
\]

We will indicate with $\overline{\tau}_\lambda:W_\lambda\rightarrow F$ the canonical projection of this subbundle of forms.

\begin{definition}[Classical Lepage-equivalent variational problem]
  The \emph{classical Lepage-equivalent variational problem associated to $\left(F,\lambda,\cI\right)$} is the variational problem $\left(W_\lambda,\Theta_\lambda,0\right)$.
\end{definition}

The following theorem allows us to write down the equations of motion for a Lepage-equivalent variational problem.

\begin{theorem}\label{Thm:HamJac} 
  For a variational problem $\left(W_\lambda,\Theta_\lambda,0\right)$ the following statements are equivalent.
  \begin{enumerate}[label=(\roman*),ref=(\roman*)]\label{}
    \renewcommand{\theenumi}{$\text{\roman{enumi}}$}
    \renewcommand{\labelenumi}{(\theenumi)}
    \renewcommand{\theenumii}{\alph{enumii}}
  \item\label{Thm1} $\sigma:M\rightarrow W_\lambda$ is critical section for the action
    \[
    \widetilde{S}\left[\sigma\right]:=\int_M\sigma^*\Theta_\lambda.
    \]
  \item\label{Thm2} $\displaystyle\sigma^*\left(Z\lrcorner\Omega_\lambda\right)=0$ for every $Z\in\mathfrak{X}^{V\left(p\circ\overline{\tau}_\lambda\right)}\left(W_\lambda\right).$
  \item\label{Thm3} $\displaystyle\sigma^*\left(Z\lrcorner\Omega_\lambda\right)=0$ for every $Z\in\mathfrak{X}\left(W_\lambda\right).$
  \end{enumerate}
\end{theorem}

It should be noted that the equations of motion of a classical Lepage-equivalent variational problem are easier to write down than the equations of motion of the original problem, because the latter involves the EDS $\cI$, that restricts in a non trivial manner the allowed sections, whereas the former is a variational problem with this EDS set to $0$. Nevertheless, there is in principle no relationship between the critical sections of these variational problems. The following result partially fills this gap.

\begin{proposition}
  Any critical section for $\left(W_\lambda,\Theta_\lambda,0\right)$ projects onto a critical section of $\left(F,\lambda,\cI\right)$.
\end{proposition}

For a proof, see \cite{GotayCartan}.

So, it remains to prove if every critical section of the original variational problem $\left(F,\lambda,\cI\right)$ can be lifted to a critical section of $\left(W_\lambda,\Theta_\lambda,0\right)$; there is no general result regarding this problem, so it is necessary to prove it in each case separately. This fact deserves a proper name.

\begin{definition}\label{def:contravariant-Var-prob}
  A Lepage-equivalent problem is \emph{contravariant} if every critical section of the original variational problem has a lift to a critical section.
\end{definition}

%%%%% Relacionar classical lepage equivalent con unified formalism!!!

\begin{remark}
  The classical Lepage-equivalent variational problem associated to the variational problem of first order field theory $\left(J^1\pi,L\nu,\cI_{\text{con}}\right)$, where $\cI_{\text{con}}$ is the contact structure on $J^1\pi$, gives us an unified formalism for first order field theory. In fact, setting $k=2$ in the above formalism, in this case we will have that $\rho\in W_{L\eta}$ if and only if
  \[
  \rho=L\nu+p_A^\mu\left(du^A-u_\nu^A dx^\nu\right)\wedge\nu_\mu,
  \]
  for some numbers $p_A^\mu$, where $\nu_\mu:=\partial_\mu\lrcorner\nu$; therefore, the map
  \[
  \rho\mapsto\left(x^\mu,u^A,p,p_A^\mu\right)
  \]
  induces coordinates on $W_{L\eta}$. Furthermore, it can be proved that
  \[
  W_{L\eta}\simeq J^1\pi\times_{E}\wedge^m_2\left(T^*E\right)
  \]
  as bundles on $E$; it is the velocity-multimomentum bundle involved in the usual unified formalism for first order field theories \cite{2004JMP....45..360E}. In this context Theorem \ref{Thm:HamJac} provides a variational formulation for the unified formalism, in the same way as \cite{doi:10.1142/S0219887815600191} does for the case of second order field theories. This formulation has been succesfully applied to the symmetry reduction of differential equations associated to variational problems, from EDS viewpoint \cite{1751-8121-45-6-065202,Morando2012}.
\end{remark}

\subsection{Geometric structures on the jet space of the frame bundle}
\label{sec:geom-struct-jet}

Before to introduce our version for the unified formalism, it is necessary to point out some interesting geometric structures associated to the jet bundle of the frame bundle of a given manifold $M$ \cite{springerlink:10.1007/PL00004852,MR0315624,brajercic04:_variat}.

In the case of the Griffiths variational problem describing Palatini gravity, the underlying bundle is $J^1\tau$, shown in the following diagram
\begin{equation}\label{eq:DiagramJetConnectionBundles}
  \begin{diagram}
    \node[2]{J^1\tau}\arrow[2]{s,l}{\tau_1}\arrow{se,t}{\tau_{10}}\arrow{sw,t}{p^{J^1\tau}_{GL\left(m\right)}}\\
    \node{C\left(LM\right)}\arrow{se,b}{\overline\tau}\node[2]{LM}\arrow{sw,b}{\tau}\\
    \node[2]{M}
  \end{diagram}
\end{equation}
The novelty in our approach rests in the fact that it uses canonical structures of the $GL\left(m\right)$-principal bundle
\[
p^{J^1\tau}_{GL\left(m\right)}:J^1\tau\rightarrow J^1\tau/GL\left(m\right)=:C\left(LM\right)
\]
on the bundle of connections $C\left(LM\right)$ for the formulation of the underlying variational problem, and thus for the description of the unified formalism. To this end, it is convenient to recall the identification
\begin{equation}\label{eq:JetToConnectionPlusFrame}
  J^1\tau\simeq C\left(LM\right)\times_M LM
\end{equation}
induced by the map
\[
j_x^1s\mapsto\left(\left[j_x^1s\right],s\left(x\right)\right).
\]
Under this diffeomorphism, the projection $\tau_{10}$ reads
\[
\tau_{10}\left(\Gamma,u\right)=u
\]
and the $GL\left(m\right)$-action becomes
\[
\left(\Gamma,u\right)\cdot g=\left(\Gamma,u\cdot g\right).
\]
According to \cite{springerlink:10.1007/PL00004852}, the projection onto the first factor
\[
\text{pr}_1:\left(\Gamma,u\right)\mapsto\Gamma
\]
gives $J^1\tau$ a $GL\left(m\right)$-principal bundle structure, and there exists a canonical\footnote{This form is essentially the contact structure of $J^1\tau$.} form $\omega\in\Omega^1\left(J^1\tau,\mathfrak{gl}\left(m\right)\right)$, which becomes a connection form in this bundle. In fact, we have that
\begin{equation}\label{eq:CanonicalConnection}
  \left.\omega\right|_{j_x^1s}=\left[T_{j_x^1s}\tau_{10}-T_xs\circ T_{j_x^1s}\tau_1\right]_{\mathfrak{gl}\left(m\right)},
\end{equation}
where $\left[\cdot\right]_{\mathfrak{gl}\left(m\right)}$ consists into the identification $V\tau\sim LM\times\mathfrak{gl}\left(m\right)$.

Using the canonical basis $\left\{E_k^l\right\}$ of $\mathfrak{gl}\left(m\right)$, such that
\[
\left(E^k_l\right)^i_j=\delta^i_l\delta_j^k,
\]
then there exists a set of $1$-forms $\left(\omega^i_j\right)$ on $J^1\tau$ such that
\[
\omega=\omega^i_jE^j_i.
\]

Also, we can pullback the $1$-form $\theta=\theta^ie_i\in\Omega^1\left(LM,\mR^m\right)$ along $\tau_{10}:J^1\tau\rightarrow LM$, obtaining a $\mR^m$-valued $1$-form on $J^1\tau$, which will be indicated with the same symbol $\theta$. This form allows us to define the \emph{canonical or universal torsion} on $J^1\tau$, according to the usual formula
\[
T^i:=d\theta^i+\omega^i_k\wedge\theta^k
\]
for the exterior covariant derivative of a tensorial form respect to the canonical connection.

For every principal connection $\Gamma$, considered as a section $\Gamma:M\rightarrow C\left(LM\right)$ of the bundle of connections, we can associate an equivariant section $\sigma_\Gamma:LM\rightarrow J^1\tau$ by means of the Diagram \eqref{eq:DiagramJetConnectionBundles} (see also Appendix \ref{App:LocalExpressions} for the explicit definition.) It allows us to construct the forms
\[
\omega_\Gamma:=\sigma_\Gamma^*\omega,\qquad T_\Gamma:=\sigma_\Gamma^*T
\]
on $LM$, having immediate geometrical interpretation.

\begin{theorem}(\cite{springerlink:10.1007/PL00004852})
  The forms $\omega_\Gamma\in\Omega^1\left(LM,\mathfrak{gl}\left(m\right)\right)$ and $T_\Gamma\in\Omega^1\left(LM,\mR^m\right)$ are the connection forms and torsion forms respectively, for the principal connection $\Gamma$.
\end{theorem}

\section{The unified formalism for Palatini gravity}
\label{sec:griff-vari-probl}

We are now ready to define an unified formalism for vacuum GR with vielbein. From the discussion of Section \ref{sec:form-hamilt-restr}, we know that a way to do it involves a two-steps procedure, namely:
\begin{enumerate}
\item To set a Griffiths variational problem for this theory, and
\item to construct the classical Lepage-equivalent variational problem for this variational problem, proving also that it is a contravariant Lepage-equivalent in the sense of Definition \ref{def:contravariant-Var-prob}.
\end{enumerate}

As we know, a Griffiths variational problem describing vacuum Palatini gravity exists \cite{capriotti14:_differ_palat}. In short, it is the variational problem specified by the triple
\[
\left(J^1\tau\rightarrow M,\cL_{PG},\cI_{PG}\right)
\]
where $\cL_{PG}=\eta^{kl}\theta_{kp}\wedge\left(d\omega^p_l+\omega^p_q\wedge\omega^q_l\right)$ and $\cI_{PG}$ is generated by the set of forms
\begin{equation}\label{eq:MetricityGenerators}
  \mathcal{G}:=\left\{\eta^{ip}\omega^j_p+\eta^{jp}\omega^i_p\right\}.
\end{equation}

By taking the subbundle $G_{PG}$ in $\wedge^m_2\left(J^1\tau\right)$ such that
\[
\left.G_{PG}\right|_{j_x^1s}:=\mR\left<\left.\alpha\right|_{j_x^1s}:\alpha\in\mathcal{G}\right>,
\]
we see that
\[
\left.W_\cL\right|_{j_x^1s}:=\left.\cL_{PG}\right|_{j_x^1s}+\left.G_{PG}\right|_{j_x^1s}
\]
for every $j_x^1s\in J^1\tau$, will define the bundle we were looking for; we will call $\overline{\tau}_\cL:W_\cL\rightarrow J^1\tau$ to the canonical projection.

\begin{remark}
  In rigor, the Griffiths variational problem we are dealing with here is not the variational problem considered in \cite{capriotti14:_differ_palat}, because we are discarding the torsion constraints in the definition of the EDS $\cI_{PG}$ adopted here. The reason for doing it is that we will ultimately find that the new variational problem also reproduces vacuum GR equations of motion (see Equation \eqref{eq:TorsionMomentaVanish} below), and in this way we are avoiding to deal with a constraint that does not live in the exterior algebra (when considered on $J^1\tau$, these constraints are equivalent to the annihilation of a set of functions on $J^1\tau$, see Lemma \ref{lem:metricity-and-torsion} below.) 
\end{remark}

Finally, we will define the $m$-form
\[
\lambda_{PG}\in\Omega^m\left(W_\cL\right)
\]
by pulling back the canonical $m$-form $\Theta_2\in\Omega^m\left(\wedge^m_2\left(J^1\tau\right)\right)$ to this subbundle.

\begin{definition}
  The \emph{unified formalism for Palatini gravity} is the variational problem associated to the set of data
  \[
  \left(W_\cL,\lambda_{PG},0\right).
  \]
\end{definition}
Therefore, from Theorem \ref{Thm:HamJac} it results that a section $\sigma:M\rightarrow W_\cL$ is a \emph{solution for this unified formalism} if and only if
\[
\sigma^*\left(Z\lrcorner d\lambda_{PG}\right)=0
\]
for any vertical vector field $Z\in\mathfrak{X}^{V\left(\tau_1\circ\overline{\tau}_\cL\right)}\left(W_\cL\right)$; it tells us that it becomes crucial to construct a basis for this set of vertical vector fields in order to find the corresponding equations of motion. This question will be considered in Section \ref{sec:convenient-basis}.

\subsection{The metricity constraints}
\label{sec:null-trace-constr}

We want to give an interpretation of the main constraints \eqref{eq:MetricityGenerators} we are imposing on a solution of the equations of motion in terms of a coordinate chart $\tau_1^{-1}\left(U\right)$ with coordinates $\left(x^\mu,e^\rho_k,e^\mu_{k\nu}\right)$.  Some additional properties related to these coordinates are listed in Appendix \ref{App:LocalExpressions}).

First, it is necessary to recall that $C\left(LM\right)$ is a affine bundle modelled on $T^*M\otimes\text{ad}\left(LM\right)$, where $\text{ad}\left(LM\right)$ is the adjoint bundle of $LM$, which is the bundle associated to the principal bundle $LM$ through the adjoint action of $GL\left(m\right)$ on itself. Thus if $\Gamma_1,\Gamma_2$ are connections at the same point $x\in M$, then we have that
\[
\Gamma_1-\Gamma_2\in T^*M\otimes\text{ad}\left(LM\right)
\]
On a coordinate open set $U\subset M$, the set of vector fields
\[
\left(x^\mu,e^\nu_k\right)\mapsto W_\nu^\mu\left(x^\mu,e^\nu_k\right):=e^\mu_ke^l_\nu\left(E^k_l\right)_{LM}\left(x^\mu,e^\nu_k\right)=-e^l_\mu\frac{\partial}{\partial e_l^\nu}
\]
generates $\text{gau}\left(LU\right)$, the set of $\tau$-vertical $GL\left(m\right)$-invariant vector fields.

Additionally, we have that the set of local sections $\Gamma\left(U,\text{ad}\left(LM\right)\right)$ can be identified with $\text{gau}\left(LU\right)$; then, every $\Gamma\in\left(C\left(LM\right)\right)_x,x\in U$ can be written as
\[
\Gamma:\frac{\partial}{\partial x^\mu}\mapsto\frac{\partial}{\partial x^\mu}+A_{\sigma\mu}^\rho\left(\Gamma\right)W_\rho^\sigma
\]
for an uniquely determined set of real numbers $\left\{A_{\sigma\mu}^\rho\left(\Gamma\right)\right\}$. The maps
\[
A_{\sigma\mu}^\rho:\Gamma\mapsto A_{\sigma\mu}^\rho\left(\Gamma\right)
\]
induces a set of coordinates in the fibres of $\left.C\left(LM\right)\right|_U\rightarrow U$.

On the other hand, the identification $\Gamma\left(U,\text{ad}\left(LM\right)\right)\simeq\text{gau}\left(LU\right)$ implies that the canonical projection $p^{TLM}_{GL\left(m\right)}:TLM\rightarrow TLM/GL\left(m\right)$ is such that
\[
p^{TLM}_{GL\left(m\right)}\circ W_\mu^\nu=W_\mu^\nu.
\]

Then if $j_x^1s=\left(x^\mu,e^\nu_k,e^\sigma_{l\rho}\right)$ in the adapted coordinates, we have that
\[
j_x^1s:\frac{\partial}{\partial x^\mu}\mapsto\frac{\partial}{\partial x^\mu}+e^\sigma_{k\mu}\frac{\partial}{\partial e^\sigma_k}=\frac{\partial}{\partial x^\mu}-e^k_\nu e^\sigma_{k\mu}W^\nu_\sigma
\]
and so
\[
\left[j_x^1s\right]:\frac{\partial}{\partial x^\mu}\mapsto\frac{\partial}{\partial x^\mu}-e^k_\nu e^\sigma_{k\mu}W^\nu_\sigma.
\]
Therefore, the projection $p^{J^1\tau}_{GL\left(m\right)}:J^1\tau\rightarrow C\left(LM\right)$ reads
\[
p_{GL\left(m\right)}^{J^1\tau}\left(x^\mu,e^\nu_k,e^\sigma_{l\rho}\right)=\left(x^\mu,-e^k_\mu e^\sigma_{k\rho}\right).
\]

\begin{remark}
  We will indicate the functions $\left(p^{J^1\tau}_{GL\left(m\right)}\right)^*A^\mu_{\nu\sigma}$ on $J^1\tau$ with the same symbol $A^\sigma_{\mu\nu}$.
\end{remark}

\begin{proposition}
  Let us denote by $\left(E_k^l\right)_{J^1\tau}\in\mathfrak{X}\left(J^1\tau\right)$ the fundamental vector field associated to $E_k^l\in\mathfrak{gl}\left(m\right)$. Then
  \[
  \left(E_k^l\right)_{J^1\tau}=\left(\left(E_k^l\right)_{LM}\right)^1=-e_k^\mu\frac{\partial}{\partial e_l^\mu}-e^\sigma_{k\mu}\frac{\partial}{\partial e^\sigma_{l\mu}}
  \]
  where $\left(\cdot\right)^1$ indicates the complete lift of a vector field to $J^1\tau$.
\end{proposition}

Let $K\subset GL\left(m\right)$ be the compact group whose Lie algebra is $\kf$. The previous identification \eqref{eq:JetToConnectionPlusFrame} allows us to establish the bundle isomorphism
\begin{equation}\label{eq:QuotientIdentifyMetrics}
  J^1\tau/K\simeq C\left(LM\right)\times_M LM/K=C\left(LM\right)\times_M\Sigma,
\end{equation}
where $\left[\tau\right]:\Sigma\rightarrow M$ is the \emph{bundle of metrics of signature $\left(m-1,1\right)$ on $M$}; therefore, we recall the following result from \cite{capriotti14:_differ_palat}.

\begin{lemma}\label{lem:metricity-and-torsion}
  The functions
  \[
  \widetilde{g}^{\mu\nu}:=\eta^{ij}e_i^\mu e_j^\nu
  \]
  are $K$-invariant, and together with the functions $A^\mu_{\nu\sigma}$ induce coordinates $\left(x^\mu,g^{\mu\nu},A^\sigma_{\mu\nu}\right)$ on $J^1\tau/K$. In term of them, the metricity condition becomes
  \[
  \eta^{ik}\omega_k^j+\eta^{jk}\omega_k^i=e^i_\mu e^j_\nu\left[d\widetilde{g}^{\mu\nu}+\left(\widetilde{g}^{\mu\sigma}A^\nu_{\gamma\sigma}+\widetilde{g}^{\nu\sigma}A^\mu_{\gamma\sigma}\right)dx^\gamma\right],
  \]
  and the universal torsion reads
  \[
  T^i=e^i_\sigma A^\sigma_{\mu\nu}dx^\mu\wedge dx^\nu.
  \]
\end{lemma}

This result permits us to give sense to the metricity condition.
\begin{theorem}\label{thm:metr-constr-interpretation}
  Let $s:M\rightarrow J^1\tau$ be a section of the jet bundle projection $\tau_1:J^1\tau\rightarrow M$, and such that it annihilates the forms $\eta^{ik}\omega_k^j+\eta^{jk}\omega_k^i$ and $T^i$. Then its quotient section $\left[s\right]:M\rightarrow J^1\tau/K$, viewed through the identification \eqref{eq:QuotientIdentifyMetrics}, consists into a metric of signature $\left(1,m-1\right)$ and the corresponding Levi-Civita connection.
\end{theorem}

\begin{remark}\label{rem:Signature-independence}
  The validity of the previous theorem is independent of the signature of the matrix $\eta$; the signature in the statement was chosen for its relationship with the usual signature found in general relativity.
\end{remark}

\begin{remark}
  The metricity form $Q^{ij}:=\eta^{ik}\omega_k^j+\eta^{jk}\omega_k^i$ is the corresponding concept in our approach to what was called \emph{nonmetricity tensor} in \cite{Friedric-1978}.
\end{remark}

\subsection{Some consequences of the structure equations}
\label{sec:structure-equations}

We will gather some identities from \cite{capriotti14:_differ_palat}, useful when dealing with differentials of the canonical forms we work with here. First, we have the \emph{structure equations}
\begin{align*}
  &d\omega^i_j+\omega^i_k\wedge\omega^k_j=\Omega^i_j\\
  &d\theta^i+\omega^i_k\wedge\theta^k=T^i,
\end{align*}
and its differential consequences, the so called \emph{Bianchi identities}
\begin{align*}
  &d\Omega^p_l=\Omega^p_r\wedge\omega^r_l-\omega^p_r\wedge\Omega^r_l\\
  &d T^k=\Omega^k_l\wedge\theta^l-\omega^k_l\wedge T^l,
\end{align*}
Additionally,
\begin{align*}
  &d\theta_i=\omega^l_i\wedge\theta_l-\omega^p_p\wedge\theta_i+T^l\wedge\theta_{li}\\
  &d\theta_{kp}=\omega^q_k\wedge\theta_{qp}-\omega^q_p\wedge\theta_{qk}-\omega^s_s\wedge\theta_{kp}+T^q\wedge\theta_{kpq}\\
  &d\theta_{ipq}=\omega^k_i \wedge\theta_{kpq}+\omega^k_p \wedge\theta_{kqi}+\omega^k_q \wedge\theta_{kip}-\omega_s^s \wedge\theta_{ipq}+T^k \wedge\theta_{ipqk}.
\end{align*}

We will use these structure equations on order to calculate the differential of the Palatini Lagrangian $\cL_{PG}$. In fact, we have that
\begin{align}\label{eq:PalatiniLagrangianDiff}
  d\cL_{PG}&=d\left(\eta^{kl}\theta_{kp}\wedge\Omega^p_l\right)\cr
  &=\eta^{kl}d\theta_{kp}\wedge\Omega^p_l+\left(-1\right)^m\eta^{kl}\theta_{kp}\wedge d\Omega^p_l\cr
  &=\eta^{kl}\left(\omega^q_k\wedge\theta_{qp}-\omega^q_p\wedge\theta_{qk}-\omega^s_s\wedge\theta_{kp}+T^q\wedge\theta_{kpq}\right)\wedge\Omega^p_l+\left(-1\right)^m\eta^{kl}\theta_{kp}\wedge\left(\Omega^p_r\wedge\omega^r_l-\omega^p_r\wedge\Omega^r_l\right)\cr
  &=\big(\eta^{kl}\omega^q_k\wedge\theta_{qp}-\eta^{kl}\omega^q_p\wedge\theta_{qk}-\eta^{kl}\omega^s_s\wedge\theta_{kp}+\eta^{kl}T^q\wedge\theta_{kpq}\big)\wedge\Omega^p_l+\cr
  &\qquad\qquad\qquad\qquad\qquad\qquad+\left(-1\right)^m\left(\eta^{kr}\theta_{kp}\wedge\omega^l_r\wedge\Omega^p_l-\eta^{kl}\theta_{kr}\wedge\omega^r_p\wedge\Omega^p_l\right)\cr
  &=\big(\eta^{kl}\omega^q_k\wedge\theta_{qp}-\eta^{kl}\omega^q_p\wedge\theta_{qk}-\eta^{kl}\omega^s_s\wedge\theta_{kp}+\eta^{kl}T^q\wedge\theta_{kpq}\big)\wedge\Omega^p_l+\cr
  &\qquad\qquad\qquad\qquad\qquad\qquad+\left(\eta^{qk}\omega^l_k\wedge\theta_{qp}\wedge\Omega^p_l-\eta^{kl}\omega^q_p\wedge\theta_{kq}\wedge\Omega^p_l\right)\cr
    &=\big[\left(\eta^{kl}\omega^q_k+\eta^{qk}\omega^l_k\right)\wedge\theta_{qp}-\eta^{kl}\omega^s_s\wedge\theta_{kp}+\eta^{kl}T^q\wedge\theta_{kpq}\big]\wedge\Omega^p_l
\end{align}

Also, there is a consequence of the first Bianchi identity that will become important later.
\begin{lemma}
  For the canonical torsion and the curvature of the canonical connection on $J^1\tau$, the following identity
  \[
  \left(d T^k+\omega^k_l\wedge T^l\right)\wedge\theta_{ikp}=\Omega^k_p\wedge\theta_{ik}-\Omega^k_i\wedge\theta_{pk}-\Omega^l_l\wedge\theta_{ip}
  \]
  holds.
\end{lemma}
\begin{proof}
  From the structure equation
  \[
  d T^k+\omega^k_l\wedge T^l=\Omega^k_l\wedge\theta^l
  \]
  and multiplying both sides by $\theta_{ikp}$, we obtain
  \[
  \left(d T^k+\omega^k_l\wedge
    T^l\right)\wedge\theta_{ikp}=\Omega^k_l\wedge\theta^l\wedge\theta_{ikp}.
  \]
  Now, we can use that
  $\theta^l\wedge\theta_{ikp}=\delta^l_p\theta_{ik}-\delta^l_k\theta_{ip}+\delta^l_i\theta_{kp}$
  and so
  \begin{align}
    \left(d T^k+\omega^k_l\wedge T^l\right)\wedge\theta_{ikp}&=\Omega^k_l\wedge\theta^l\wedge\theta_{ikp}\cr
    &=\Omega^k_l\wedge\left(\delta^l_p\theta_{ik}-\delta^l_k\theta_{ip}+\delta^l_i\theta_{kp}\right)\cr
    &=\Omega^k_p\wedge\theta_{ik}-\Omega^k_i\wedge\theta_{pk}-\Omega^l_l\wedge\theta_{ip}.\qedhere
  \end{align}
\end{proof}

The following consequence will be useful when dealing with the involutivity of the equations of motion (see Section \ref{sec:constraint-algorithm}).

\begin{corollary}\label{cor:FirstBianchiConsequence}
  If $T=0$ and $\Omega$ is $\kf$-valued, then
  \[
  \Omega^k_p\wedge\theta_{ik}=\Omega^k_i\wedge\theta_{pk}.
  \]
\end{corollary}

\section{A convenient basis}
\label{sec:convenient-basis}

According to Theorem \ref{Thm:HamJac}, in order to be able to write down the equations of motion for the unified formalism just constructed, it is necessary to contract arbitrary vertical vector fields with the differential of the form $\lambda_{PG}$ on $W_\cL$. In this section we will define a basis for $TJ^1\tau$ that will allow us, in Section \ref{sec:equations-motion}, to carry out this procedure. The existence of this basis is tied to the choice of a connection on the frame bundle $LM$; nevertheless, the fact that in the calculation of the equations just vertical vector fields are involved, determines that these equations are independent from this choice.

The construction of this basis goes as follows: The chosen connection is employed in the definition of a basis on $LM$. After that, using the canonical lifts and other tools at our disposal in any jet bundle \cite{saunders89:_geomet_jet_bundl}, we will be able to find a basis on $J^1\tau$. Finally, and using an identification of the velocity-multimomentum bundle with a product bundle involving $J^1\tau$ (see Section \ref{sec:veloc-mult-space-1}), it will be possible to use this basis in the calculation of the equations of motion.

A word of caution should be said here: as we said before, the equations of motion for the unified formalism can be written using only the vertical part of the basis to be built in the present section. However, we feel that the construction of a full basis on $J^1\tau$ could be of interested for the readers, even if it makes the article a little longer and harder to read.

\subsection{Basis on $LM$}
\label{sec:basis-lm}

Let us fix a linear connection $\sigma\in\Omega^1\left(LM,\mathfrak{gl}\left(m\right)\right)$; also, we must recall the construction of the \emph{standard horizontal vector field corresponding to }$\xi\in\mR^n$ \cite{KN1}. This vector field $B\left(\xi\right)$ is characterized by two properties:
\begin{itemize}
\item It is a horizontal vector field, and
\item $T_u\tau\left(B\left(\xi\right)\right)=u\left(\xi\right)$ for every $u\in LM$.
\end{itemize}

Moreover, let $\left\{E_i^j:i,j=1,\cdots,m\right\}$ be a basis for $\mathfrak{gl}\left(m\right)$ given by the formula
\[
\left(E_k^l\right)_i^j=\delta^l_i\delta^j_k;
\]
as we said above, symbols $\left(E^j_i\right)_{LM}\in\mathfrak{X}\left(LM\right)$ refer to the infinitesimal generator associated to the element $E^j_i$. The set $\left\{B\left(e^i\right),\left(E_i^j\right)_{LM}\right\}$ results to be a basis of vector fields on $LM$.

\subsection{Adapted coordinates calculations}
\label{sec:adapt-coord-calc}

Let us calculate the previous vector fields in an adapted coordinate chart $\left(x^\mu,e^\mu_i\right)$ on $LM$ \cite{KN1,brajercic04:_variat}. The connection form in these coordinates becomes
\[
\sigma=e^i_\mu\left(de^\mu_j+\Xi^\mu_{\sigma\nu}e^\nu_jdx^\sigma\right)E_i^j
\]
for some functions $\left\{\Xi^\mu_{\sigma\nu}\right\}$ defined on the domain of the coordinates $\left(x^\mu\right)$. Therefore, if
\[
B\left(e_i\right):=B^\mu_i\frac{\partial}{\partial x^\mu}+B^\mu_{ik}\frac{\partial}{\partial e^\mu_k},
\]
projectability condition, namely
\[
T_u\tau\left(B\left(e_i\right)\right)=u\left(e_i\right),
\]
reduces to
\[
B^\mu_i=e^\mu_i
\]
for all $\mu,i=1,\cdots,m$. On the other hand, horizontality condition $\sigma\left(B\left(e_i\right)\right)=0$ becomes
\begin{align*}
  0&=\sigma\left(B\left(e_i\right)\right)\\
  &=e_\mu^l\left(B_{ij}^\mu+\Xi^\mu_{\sigma\nu}e^\nu_je^\sigma_i\right);
\end{align*}
it means that the standard horizontal vector fields have the local expressions
\[
B\left(e_i\right)=e_i^\mu\left(\frac{\partial}{\partial x^\mu}-e^\nu_j\Xi^\sigma_{\mu\nu}\frac{\partial}{\partial e^\sigma_j}\right).
\]

Also, in the above coordinates every element $M:=\left(M_k^l\right)\in GL\left(m\right)$ acts according to the formula
\[
\left(x^\mu,e^\rho_k\right)\cdot M=\left(x^\mu,e^\rho_kM^k_l\right).
\]
Therefore, if $E_k^l\in\mathfrak{gl}\left(m\right)$ is such that
\[
\left(E_k^l\right)_i^j=\delta^l_i\delta^j_k,
\]
then its associated fundamental vector field reads
\begin{align*}
  \left(E_k^l\right)_{LM}\left(x^\mu,e^\rho_k\right)&=\left.\frac{\vec{\text{d}}}{\text{d}t}\right|_{t=0}\left(x^\mu,e^\rho_i\left(\exp{\left(-tE_k^l\right)}\right)^i_j\right)\\
  &=-e^\rho_i\left(E_k^l\right)^i_j\frac{\partial}{\partial e^\rho_j}\\
  &=-e^\rho_k\frac{\partial}{\partial e^\rho_l}.
\end{align*}

\subsection{Lifts to $J^1\tau$}
\label{sec:lifts-j1tau}

Given a vector field $X$ on $LM$, we can construct a vector field $X^1$ on $J^1\tau$, called \emph{prolongation} \cite{saunders89:_geomet_jet_bundl}, by differentiating the jet lift of the flow $\phi^X_t:LM\rightarrow LM,t\in\left(-\epsilon,\epsilon\right)$ associated to $X$, namely
\[
X^1\left(j_x^1s\right):=\left.\frac{\vec{\text{d}}}{\text{d}t}\right|_{t=0}\left[j^1_x\left(\Phi^X_t\circ s\right)\right].
\]
In adapted coordinates $\left(x^\mu,e^\mu_k,e^\mu_{k\nu}\right)$, we have that if
\[
X=X^\mu\frac{\partial}{\partial x^\mu}+X^\mu_k\frac{\partial}{\partial e^\mu_k}
\]
then
\begin{equation}\label{eq:ProlongationVectorFormula}
  X^1=X^\mu\frac{\partial}{\partial x^\mu}+X^\mu_k\frac{\partial}{\partial e^\mu_k}+\left(D_\kappa X^\nu_j-e_{j\sigma}^\nu D_\kappa X^\sigma\right)\frac{\partial}{\partial e_{j\kappa}^\nu},
\end{equation}
where
\[
D_\kappa F:=\frac{\partial F}{\partial x^\kappa}+e_{k\kappa}^\mu\frac{\partial F}{\partial e^\mu_k}
\]
is the total derivative operator on $J^1\tau$.

It is also necessary to consider \emph{vertical lifts} for a vertical vector field $X\in\mathfrak{X}^{V\tau}\left(LM\right)$; nevertheless, in order to construct such a lift we must have at our disposal a form $\alpha\in\Omega^1\left(M\right)$. Then we use the affine structure of $J^1\tau$, which is modelled on the vector bundle $\tau^*T^*M\otimes V\tau$, in order to define the vertical lift according to the formula
\[
\left(\alpha,X\right)^V\left(j_x^1s\right):=\left.\frac{\vec{\text{d}}}{\text{d}t}\right|_{t=0}\left[j_x^1s+t\alpha_x\otimes X\left(s\left(x\right)\right)\right].
\]
In adapted coordinates, if $\alpha:=\alpha_\nu dx^\nu$ and $X$ is as before, then
\[
\left(\alpha,X\right)^V=\alpha_\nu X^\mu_k\frac{\partial}{\partial e^\mu_{k\nu}}.
\]

Now, a consequence of being working in $LM$ is that for every $j_x^1s\in J^1\tau$ there exists a $\mR^m$-valued $1$-form $\theta=\theta^ie_i$; therefore, we can generalise the previous vertical lift construction by using these forms, namely
\[
\left(\theta^i,X\right)^V\left(j_x^1s\right):=\left.\frac{\vec{\text{d}}}{\text{d}t}\right|_{t=0}\left[j_x^1s+t\left.\theta^i\right|_{j_x^1s}\otimes X\left(s\left(x\right)\right)\right].
\]
In local terms it becomes
\begin{equation}\label{eq:VerticalLiftTheta}
  \left(\theta^i,X\right)^V=e_\nu^iX^\mu_k\frac{\partial}{\partial e^\mu_{k\nu}}.
\end{equation}

The brackets between these lifts have a particular structure.

\begin{lemma}\label{lem:LiftsBrackets}
  Let $\pi:E\rightarrow M$ be a bundle and $X,Y\in\mathfrak{X}\left(E\right)$ a pair of vector fields on the bundle; let $\alpha,\beta\in\Omega^1\left(M\right)$ be $1$-forms on $M$. Then
  \begin{itemize}
  \item $\displaystyle\left[X^1,Y^1\right]=\left(\left[X,Y\right]\right)^1$.
  \item $\displaystyle\left[X^1,\left(\alpha,Y\right)^V\right]=\left(\alpha,\left[X,Y\right]\right)^V$.
  \item $\displaystyle\left[\left(\alpha,X\right)^V,\left(\beta,Y\right)^V\right]=0$.
  \end{itemize}
\end{lemma}
% \begin{proof}
%   Let $\Phi^X_t,\Phi^Y_t:E\rightarrow E$ be the flows of these vector fields; we have the formula
%   \[
%   \left[X,Y\right]:=\left.\frac{\vec{d}}{dt}\right|_{t=0}\left(\Phi^Y_{-\sqrt{t}}\circ\Phi^X_{-\sqrt{t}}\circ\Phi^Y_{\sqrt{t}}\circ\Phi^X_{\sqrt{t}}\right).
%   \]
%   Therefore
%   \begin{align*}
%     \left[X^1,Y^1\right]&:=\left.\frac{\vec{d}}{dt}\right|_{t=0}\left(j^1\Phi^Y_{-\sqrt{t}}\circ j^1\Phi^X_{-\sqrt{t}}\circ j^1\Phi^Y_{\sqrt{t}}\circ j^1\Phi^X_{\sqrt{t}}\right)\\
%     &=\left.\frac{\vec{d}}{dt}\right|_{t=0}\left[j^1\left(\Phi^Y_{-\sqrt{t}}\circ\Phi^X_{-\sqrt{t}}\circ\Phi^Y_{\sqrt{t}}\circ\Phi^X_{\sqrt{t}}\right)\right]\\
%     &=\left(\left[X,Y\right]\right)^1
%   \end{align*}
%   as required.
% \end{proof}

\subsection{Local expressions for the lifts of a particular basis}
\label{sec:local-expr-lift}

Now, we will proceed to calculate the local expressions for the lifts associated to the basis $\left\{B\left(e_i\right),\left(E_j^k\right)_{LM}\right\}$ on $LM$. Using the calculations made in Section \ref{sec:adapt-coord-calc}, we will have that
\begin{align*}
  D_\kappa\left(B\left(e_i\right)\right)^\mu&=\frac{\partial\left(B\left(e_i\right)\right)^\mu}{\partial x^\kappa}+e_{l\kappa}^\sigma\frac{\partial\left(B\left(e_i\right)\right)^\mu}{\partial e^\sigma_l}\\
  &=e^\mu_{i\kappa}
\end{align*}
and
\begin{align*}
  D_\kappa\left(B\left(e_i\right)^\mu_j\right)&=\frac{\partial}{\partial x^\mu}\left(-\Xi^\mu_{\sigma\nu}e^\sigma_ie^\nu_j\right)+e^\rho_{l\kappa}\frac{\partial}{\partial x^\rho_l}\left(-\Xi^\mu_{\sigma\nu}e^\sigma_ie^\nu_j\right)\\
  &=-e^\sigma_ie^\nu_j\frac{\partial\Xi^\mu_{\sigma\nu}}{\partial x^\kappa}-\Xi^\mu_{\sigma\nu}\left(e^\sigma_{i\kappa}e^\nu_j+e^\sigma_ie^\nu_{j\kappa}\right);
\end{align*}
therefore
\[
\left(B\left(e_i\right)\right)^1=e_i^\mu\left(\frac{\partial}{\partial x^\mu}-e^\nu_j\Xi^\sigma_{\mu\nu}\frac{\partial}{\partial e^\sigma_j}\right)-\left[e^\sigma_ie^\nu_j\frac{\partial\Xi^\mu_{\sigma\nu}}{\partial x^\kappa}+\Xi^\mu_{\sigma\nu}\left(e^\sigma_{i\kappa}e^\nu_j+e^\sigma_ie^\nu_{j\kappa}\right)+e^\mu_{j\sigma}e^\sigma_{i\kappa}\right]\frac{\partial}{\partial e^\mu_{j\kappa}}.
\]
In the same vein
\[
\left(\left(E_k^l\right)_{LM}\right)^1=-e^\rho_k\frac{\partial}{\partial e^\rho_l}-e^\rho_{k\kappa}\frac{\partial}{\partial e^\rho_{l\kappa}}=\left(E^l_k\right)_{J^1\pi}.
\]

For the vertical lifts of the vertical vector fields $\left(E^l_k\right)_{LM}$, we use formula \eqref{eq:VerticalLiftTheta}, obtaining
\[
\left(\theta^i,\left(E^l_k\right)_{LM}\right)^V=-e^i_\nu e^\mu_k\frac{\partial}{\partial e^\mu_{l\nu}}.
\]

\subsection{Basis on $J^1\tau$}
\label{sec:basis-j1tau}

According to the discussion carried out in the previous paragraphs, we will use as a basis for writing out the equations of motion of Palatini gravity the following set of (global!) vector fields
\[
B:=\left\{\left(B\left(e_i\right)\right)^1,\left(E^k_l\right)_{J^1\tau},\left(\theta^i,E^j_k\right)^V:i,j,k,l=1,\cdots,m\right\}.
\]
Also, we will need the dual basis $B^*$ corresponding to $B$; in order to find it, recall that $\theta,\omega$ are semibasic respect to the projection $\tau_{10}:J^1\tau\rightarrow LM$, so
\[
\left(\theta^i,E^j_k\right)^V\lrcorner\theta^l=0=\left(\theta^i,E^j_k\right)^V\lrcorner\omega^p_q.
\]
Also, $\theta$ is semibasic respect to the projection $\tau_1:J^1\tau\rightarrow M$, so we will have that
\[
\left(E^k_l\right)_{J^1\tau}\lrcorner\theta^i=0.
\]
The infinitesimal generators for the $GL\left(m\right)$-action also have the useful property
\[
\left(E^k_l\right)_{J^1\tau}\lrcorner\omega^p_q=\delta^p_l\delta^k_q.
\]
To this list of niceties, we would add Lemma \ref{lem:CVerticalDuality}, that uses the difference map $C$ defined in Appendix \ref{sec:diff-tens-c_i}, and tells us that the exact forms
\[
\Psi^k_{jl}:=dC^k_{jl}
\]
could be used in the search of dual directions corresponding to the set of vertical vector fields $\left\{\left(\theta^i,\left(E^k_l\right)_{LM}\right)\right)^V$. 

Please note that the contraction of these forms with the infinitesimal generators of the $GL\left(m\right)$-action on $J^1\tau$ can be calculated using using Corollary \ref{cor:VerticalDerivativesFunctonDifference}; in fact, from there we conclude that
\begin{align}
  \left(A_{LM}\right)^1\cdot C^k_{li}&=\left(A_{J^1\tau}\cdot C\right)^k_{li}\cr
  &=\left(-\left[A,C\right]+A^*C\right)^k_{li}\cr
  &=-A^p_lC^k_{pi}+C^p_{li}A_p^k+A^p_iC_{lp}^k,\label{eq:InfGenOnC}
\end{align}
or, using the basis $\left\{E_q^p\right\}$ for the Lie algebra $\mathfrak{gl}\left(m\right)$,
\[
\left(E^p_q\right)_{J^1\tau}\cdot C^k_{li}=-\delta^p_lC^k_{qi}+\delta^k_qC^p_{li}+\delta^p_iC^k_{lq}=:F^{pk}_{qli}.
\]

On the contrary, ``horizontal'' vector fields $\left(B\left(e_i\right)\right)^1$ does not share these characteristics with the rest of vector fields in $B$; namely, the contraction with $\omega$ gives rise to the difference function $C$ (see Appendix \ref{sec:diff-tens-c_i})
\[
\left(B\left(e_i\right)\right)^1\lrcorner\omega_k^l=C_{ki}^l
\]
and the contraction with the forms $\Psi^k_{jl}$ yields to a new set of functions on $J^1\tau$, namely
\[
D^l_{jki}:=\left(B\left(e_i\right)\right)^1\cdot C^l_{jk}.
\]

\begin{remark}
  The fact that these functions cannot be described in terms of $C$ is
  consequence of the proposition we will formulate below; in short,
  the functions $D$ encode information about the covariant derivative
  of $C$ respect to the connection $\sigma$, if we consider $C$ as a
  section $\phi_C:C\left(LM\right)\rightarrow E$ of the bundle
  $\pi:E\rightarrow C\left(LM\right)$ associated to
  $J^1\tau\rightarrow C\left(LM\right)$ through
  $\left(\mathfrak{gl}\left(m\right)\otimes\left(\mR^m\right)^*,\rho\right)$
  (see Remark \ref{rem:CAsASectionOfE}). Now, let
  $\pi':E'\rightarrow M$ be the fibre bundle associated to $LM$
  through the same representation, and consider the product bundle
  \[
  F:=C\left(LM\right)\otimes_M E'.
  \]
  Recall the isomorphism
  \[
  J^1\tau\simeq C\left(LM\right)\times_M LM
  \]
  given by
  $j_x^1s\mapsto\left(\left[j_x^1s\right],s\left(x\right)\right)$, and
  consider the bundle $E$ as a bundle on $M$ through the composite map
  \[
  \overline{\pi}:=\overline{\tau}\circ\pi:E\rightarrow M;
  \]
  then there exists a bundle map isomorphism over the identity on $M$
  \begin{center}
    \begin{tikzcd}
      E \arrow[rr, "f"] \arrow[dr, "\overline{\pi}",swap] & & F
      \arrow[dl, "\overline{\pi}'"]
      \\
      & M
    \end{tikzcd}
  \end{center}
  It is given by the formula
  \[
  f:E\rightarrow
  F:\left[j_x^1s,A\otimes\alpha\right]\mapsto\left(\left[j_x^1s\right],\left[s\left(x\right),A\otimes\alpha\right]\right);
  \]
  thus, a connection $\Gamma$ on $LM$, considered as a section
  $\sigma_\Gamma:M\rightarrow C\left(LM\right)$, can be used together
  with the section $\phi_C:C\left(LM\right)\rightarrow E$ in order to
  define a new section
  \[
  \phi'_C:M\rightarrow F:x\mapsto
  f\left(\phi_C\left(\sigma_\gamma\left(x\right)\right)\right).
  \]

\begin{proposition}
  Let $s:M\rightarrow LM$ be a section of the bundle of frames; let
  \[
  X_i\left(x\right):=s\left(x\right)\left(e_i\right)
  \]
  be the basis of vector fields on $M$ determined by the section $s$
  and the basis $\left\{e_i\right\}\subset\mR^m$. Then
  \[
  \nabla_{X_i}\phi'_C=D^l_{jki}E^j_l\otimes e^k,
  \]
  where $\nabla$ is the covariant derivative on $F$ associated to the
  connection $\sigma$.
\end{proposition}
\end{remark}

\subsection{An almost dual basis on $J^1\tau$}
\label{sec:almost-dual-basis}

Now, we can define the following $1$-forms on $J^1\tau$
\begin{align*}
  \Psi^i_{jk}&:=dC^i_{jk}-F^{pi}_{qjk}\omega^q_p-\left(D^i_{jkp}-F^{si}_{rjk}C^r_{sp}\right)\theta^p\\
  \rho^k_l&:=\omega^k_l-C^k_{lp}\theta^p;
\end{align*}
recalling the calculations performed in Appendix \ref{sec:contr-elem-b}, the basis
\[
B^*:=\left\{\theta^i,\rho^k_l,\Psi^i_{jk}\right\}\subset\Omega^1\left(J^1\tau\right)
\]
becomes the dual basis for $B$. With these elements we will be able to write down the contraction of $\Omega$ and $T$ with the elements of $B$.

\begin{lemma}
  The contractions of $\Omega$ and $T$ with the elements of $B$ become
  \begin{align*}
    &\left(B\left(e_i\right)\right)^1\lrcorner\Omega^k_l=\left[D^k_{lij}-D^k_{lji}+R^k_l\left(B\left(e_j\right),B\left(e_i\right)\right)+C^k_{pj}C^p_{li}-C^k_{pi}C^p_{lj}\right]\theta^j+\Psi^k_{li},\\
    &A_{J^1\tau}\lrcorner\Omega^k_l=0,\\
    &\left(\theta^r,\left(E^k_l\right)_{LM}\right)^V\lrcorner\Omega^q_p=-\delta_p^k\delta_l^q\theta^r,\\
    &\left(B\left(e_i\right)\right)^1\lrcorner T^k=\left(C_{ij}^k-C_{ji}^k\right)\theta^j,\\
    &A_{J^1\tau}\lrcorner T^k=0,\\
    &\left(\theta^r,\left(E^k_l\right)_{LM}\right)^V\lrcorner T^q=0.
  \end{align*}
  for any $A\in\mathfrak{gl}\left(m\right)$.
\end{lemma}

\section{The velocity-multimomentum space $W_\cL$ and its canonical form}
\label{sec:veloc-mult-space}

As we explained at the beginning of the previous section, it is time to write down the equations of motion associated to our unified formulation of Palatini gravity, as it is was formulated in Section \ref{sec:griff-vari-probl}. Before that, it will be necessary to simplify somewhat the bundle $W_\cL$; it will be achieved achieved by taking into account that the restriction EDS $\cI_{\text{PG}}$ has global generators. As a consequence, $W_\cL$ will become isomorphic to a product bundle on $J^1\tau$, and so it will be possible to lift the elements of the basis $B$ in order to be part of a basis of $W_\cL$.

\subsection{The velocity-multimomentum space}
\label{sec:veloc-mult-space-1}

Recall that any element $\rho\in\overline{\tau}_\cL^{-1}\left(j_x^1s\right)\subset W_\cL$ can be written as
\[
\rho=\eta^{kl}\theta_{kp}\wedge\left(d\omega^p_l+\omega^p_q\wedge\omega^q_l\right)+\eta^{ql}\beta_{pq}\wedge\omega^p_l,
\]
where $\beta_{kl}\in\wedge^{m-1}\left(T_x^*M\right)$ and $\beta_{kl}-\beta_{lk}=0$.

Nevertheless, it would be preferable to keep the basic symmetries of the bundles involved in this description; for that reason, we will take $\beta_{ij}$ as elements on a bundle whose sections have special transformations properties regarding the $GL\left(m\right)$-action on $LM$. Namely, let us consider the natural representation of $GL\left(m\right)$ on $S^*\left(m\right):=\left(\mR^m\right)^*\odot\left(\mR^m\right)^*$, where $\odot$ indicates the symmetrised tensorial product. Then the bundle
\begin{center}
  \begin{tikzcd}[column sep=1.7cm, row sep=1.5cm]
    E_2:=\wedge^{m-1}_1\left(J^1\tau\right)\otimes S^*\left(m\right)
    \arrow{r}{p_2}
    &
    J^1\tau
  \end{tikzcd}
\end{center}
where
\[
\wedge^k_1\left(J^1\tau\right):=\left\{\gamma\in\wedge^k\left(J^1\tau\right):\gamma\text{ is horizontal respect to the projection }\tau_1:J^1\tau\rightarrow M\right\},
\]
will provide the forms $\beta_{ij}$. The considerations made in Appendix \ref{sec:canonical-k-form} apply to this bundle: We must take $P=J^1\tau, N=C\left(LM\right)$ and $G=GL\left(m\right)$. Then it is a $GL\left(m\right)$-space through the action
\[
\Phi_g\left(\gamma\right):=g\cdot\left(\alpha\circ\left(T\phi_{g}^1\times\cdots\times T\phi_{g}^1\right)\right)
\]
where $\phi^1_g:J^1\tau\rightarrow J^1\tau$ is the lift of the action $R_{g^{-1}}:LM\rightarrow LM$ for every $g\in GL\left(m\right)$. In particular, the canonical $m-1$-form on $E_2$ is a $GL\left(m\right)$-equivariant form.

\begin{lemma}\label{lem:WLIsomorphism}
  The map
  \[
  \Gamma:\rho\mapsto\left(j_x^1s,\beta_{ij}e^i\odot e^j\right)
  \]
  induces an isomorphism of bundles
  \[
  W_\cL\simeq J^1\tau\times_ME_2,
  \]
  making the following diagram commutative
  \begin{center}
    \begin{tikzcd}[ampersand replacement=\&]
      W_\cL 
      \arrow[rr, "\Gamma"]
      \arrow[dr, "\overline{\tau}_\cL",swap]
      \&
      \&
      J^1\tau\times_ME_2
      \arrow[dl ,"\text{pr}_0"]
      \\
      \&
      J^1\tau
    \end{tikzcd}
  \end{center}
  Moreover, this map is equivariant respect to the $GL\left(m\right)$-action in each of these spaces.
\end{lemma}

\subsection{The canonical form}
\label{sec:canonical-form}

From now on we will introduce the notation
\[
\widehat{W_\cL}:=J^1\tau\times_ME_2.
\]
With this result in mind, we can give a formula for the canonical $m$-form on $W_\cL$; the idea is to lift the forms $\omega,\theta,T,\dots$ from $J^1\tau$ to $\widehat{W_\cL}$ through $\text{pr}_0:\widehat{W_\cL}\rightarrow J^1\tau$, and to do the same to the canonical form on $E_2$, using the projection
\[
\text{pr}_2:\widehat{W_\cL}\longrightarrow E_2.
\]
Whenever possible, we will indicate the forms lifted from $J^1\tau$ with the same symbols that the forms that are being lifted (i.e. $\overline{\tau}_\cL^*\omega^k_l\rightsquigarrow\omega^k_l$ and so on) and the components\footnote{Respect to the canonical basis of $S^*\left(m\right)$.} of the canonical $m-1$-form, both the original and the lifted, will be indicated by $\Theta_{ij}$. With these notational conventions the pullback $\lambda_\cL$ of the canonical $m$-form from $\wedge^m\left(J^1\tau\right)$ to $W_\cL$ at the point
\[
\rho=\eta^{kl}\theta_{kp}\wedge\left(d\omega^p_l+\omega^p_q\wedge\omega^q_l\right)+\eta^{ql}\beta_{pq}\wedge\omega^p_l
\]
will read
\begin{equation}\label{eq:LambdaInTermsCanForms0}
  \left.\lambda_{\cL}\right|_\rho=\eta^{kl}\theta_{kp}\wedge\left(d\omega^p_l+\omega^p_q\wedge\omega^q_l\right)+\eta^{ql}\left.\Theta_{pq}\right|_{\beta}\wedge\omega^p_l.
\end{equation}
Using Equation \eqref{eq:PalatiniLagrangianDiff} from Section \ref{sec:structure-equations}, we have that
\begin{multline*}
  \left.d\lambda_\cL\right|_\rho=\left[\left(\eta^{kp}\omega_k^i+\eta^{ki}\omega_k^p\right)\wedge\theta_{il}-\omega^s_s\wedge\eta^{kp}\theta_{kl}+\eta^{kp}T^i\wedge\theta_{kli}+\eta^{ip}\left.\Theta_{il}\right|_{\beta}\right]\wedge\Omega^l_p+\\
  +\left.d\Theta_{ij}\right|_{\beta}\wedge\left(\eta^{ik}\omega_k^j+\eta^{jk}\omega_k^i\right)-\eta^{ik}\left.\Theta_{ij}\right|_{\beta}\wedge\omega^j_l\wedge\omega_k^l.
\end{multline*}
Now, recalling the discusion of Appendix \ref{sec:cart-decomp-forms}, we have that
\[
\eta^{ik}\omega_k^j+\eta^{jk}\omega_k^i=\eta^{ik}\left(\omega^j_k+\eta_{kr}\omega^r_l\eta^{lj}\right)=2\eta^{ik}\left(\omega_\pf\right)_k^j,\qquad\omega^s_s=\left(\omega_\pf\right)^s_s.
\]
Also we have that
\[
\eta^{ik}\Theta_{kj}E_i^j\in\pf,
\]
because $\Theta_{ij}+\Theta_{ji}=0$; therefore, from Proposition \ref{prop:TraceWedge}, it results that
\[
\eta^{ik}\left.\Theta_{ij}\right|_{\beta}\wedge\omega^j_l\wedge\omega_k^l=\eta^{ik}\left.\Theta_{ij}\right|_{\beta}\wedge\left[\left(\omega\wedge\omega\right)_\pf\right]_j^i,
\]
and so, using Proposition \ref{prop:DecompSquare}, we obtain
\[
\eta^{ik}\left.\Theta_{ij}\right|_{\beta}\wedge\omega^j_l\wedge\omega_k^l=\eta^{ik}\left.\Theta_{ij}\right|_{\beta}\wedge\left[\left(\omega_\pf\right)^j_p\wedge\left(\omega_\kf\right)^p_k+\left(\omega_\kf\right)^j_p\wedge\left(\omega_\pf\right)^p_k\right].
\]
Replacing these identities in the previous expression for $d\lambda_\cL$, it becomes
\begin{multline*}
  \left.d\lambda_\cL\right|_\rho=\left[2\eta^{kp}\left(\omega_\pf\right)_k^i\wedge\theta_{il}-\left(\omega_\pf\right)^s_s\wedge\eta^{kp}\theta_{kl}+\eta^{kp}T^i\wedge\theta_{kli}+\eta^{ip}\left.\Theta_{il}\right|_{\beta}\right]\wedge\Omega^l_p+\\
  +\eta^{ik}\left.d\Theta_{ij}\right|_{\beta}\wedge\left(\omega_\pf\right)_k^j-\eta^{ik}\left.\Theta_{ij}\right|_{\beta}\wedge\left[\left(\omega_\pf\right)^j_p\wedge\left(\omega_\kf\right)^p_k+\left(\omega_\kf\right)^j_p\wedge\left(\omega_\pf\right)^p_k\right].
\end{multline*}

It is interesting to rearrange some terms, and put them in the following form
\begin{multline}\label{eq:FormulaFordLambda0}
  \left.d\lambda_\cL\right|_\rho=\left[2\eta^{kp}\left(\omega_\pf\right)_k^i\wedge\theta_{il}-\left(\omega_\pf\right)^s_s\wedge\eta^{kp}\theta_{kl}+\eta^{kp}T^i\wedge\theta_{kli}+\eta^{ip}\left.\Theta_{il}\right|_{\beta}\right]\wedge\Omega^l_p+\\
  +\eta^{ik}\left[\left.d\Theta_{ij}\right|_{\beta}+\eta^{rq}\eta_{li}\left.\Theta_{rj}\right|_{\beta}\wedge\left(\omega_\kf\right)^l_q-\left.\Theta_{ip}\right|_{\beta}\wedge\left(\omega_\kf\right)^p_j\right]\wedge\left(\omega_\pf\right)^j_k.
\end{multline}

\begin{remark}
  We will note that the canonical form on $E_2$
  is $S^*\left(m\right)$-valued, and that
  it is also $GL\left(m\right)$-equivariant; thus let us define
  \begin{align*}
    D\Theta^{m-1}&:=\left(d\Theta_{ij}+\eta^{rk}\eta_{iq}\Theta_{rj}\wedge\omega^q_k-\Theta_{ip}\wedge\omega^p_j\right)e^i\odot e^j.
  \end{align*}
  Then the tautological property of the canonical forms allows us to set the following result.
  
  \begin{lemma}
    For
    $\beta\in\Omega^{m-1}\left(E_2\right)$,
    we have that
    \[
    D\beta=\beta^*\left(D\Theta^{m-1}\right)
    \]
    is the exterior covariant differential of the form
    $\beta$ respect to the canonical connection $\omega$.
  \end{lemma}
  
  This lemma gives us a hint on the interpretation of some
  terms present in Equation \eqref{eq:FormulaFordLambda0} for $d\lambda_\cL$.
\end{remark}

\subsection{Equations of motion}
\label{sec:equations-motion}

From Theorem \ref{Thm:HamJac} we know that the equations of motion arise from 
\[
Z\lrcorner d\lambda_\cL=0
\]
for $Z\in\mathfrak{X}^{V\left(\tau_1\circ\overline{\tau}_\cL\right)}\left(W_\cL\right)$. On the other hand, Lemma \ref{lem:WLIsomorphism} allows us to conclude that
\[
V{\left(\tau_1\circ\overline{\tau}_\cL\right)}\simeq V\tau_1\oplus V\left(\tau_1\circ p_2\right);
\]
it means that a set of generators for $\mathfrak{X}^{V\left(\tau_1\circ\overline{\tau}_\cL\right)}\left(\widehat{W}_\cL\right)$ can be constructed with vector fields such as
\[
X+0,\qquad 0+Y
\]
where $X\in\mathfrak{X}^{V\tau_1}\left(J^1\tau\right),Y\in\mathfrak{X}^{Vp_2}\left(E_2\right)$. Now, from Section \ref{sec:basis-j1tau} we know that
\[
\left\{\left(E^k_l\right)_{J^1\tau},\left(\theta^i,E^j_k\right)^V:i,j,k,l=1,\cdots,m\right\}
\]
is a set of generators for the set of $\tau_1$-vertical vector fields, and any section $\beta\in\Gamma\left(E_2\right)$ gives rise to a $p_2$-vertical vector field
\[
\delta\beta\left(\alpha_{j_x^1s}\right):=\left.\frac{d}{dt}\right|_{t=0}\left[\alpha_{j_x^1s}+t\beta\left(j_x^1s\right)\right],\qquad\alpha_{j_x^1s}\in E_2
\]
on $E_2$, and the collection of this sort of vector fields is a set of generators for $p_2$-vertical vector fields.

Thus Equation \eqref{eq:FormulaFordLambda0} will yield to the equations of motion, by contracting this form with the set of vector fields
\[
A_{J^1\tau}+0,\qquad\left(\theta^i,B_{LM}\right)^V+0,\qquad 0+\delta\beta
\]
for $A,B\in\mathfrak{gl}\left(m\right)$ and $\beta\in\Gamma\left(E_2\right)$. In fact:
\begin{itemize}
\item Contraction with elements of the form $0+\delta\beta$ gives us the equations
  \[
  \eta^{ik}\delta\beta_{ij}\left(\omega_{\pf}\right)^j_k=0;
  \]
  it means that
  \begin{equation}\label{eq:MetricConditionEq}
    \left(\omega_{\pf}\right)^j_k=0,
  \end{equation}
that is, the connection must be metric.
\item Contraction with vector fields of the form $\left(\theta^r,\left(E^s_t\right)_{LM}\right)^V$ will give us
  \begin{align*}
    0&=\left[2\eta^{kp}\left(\omega_\pf\right)^i_k\wedge\theta_{il}-\left(\omega_\pf\right)^i_i\eta^{kp}\theta_{kl}+\eta^{kp}T^i\wedge\theta_{kli}+\eta^{ip}\Theta_{il}\right]\wedge\delta^s_p\delta^l_t\theta^r\\
    &=\eta^{ks}\left(T^i\wedge\theta_{kti}+\Theta_{kt}\right)\wedge\theta^r,
  \end{align*}
  where in the last step Equation \eqref{eq:MetricConditionEq} was used. It is equivalent to equation
  \begin{equation}\label{eq:EqTorsionTheta}
    T^i\wedge\theta_{kti}+\Theta_{kt}=0,
  \end{equation}
and this is turn means that
  \[
  \Theta_{kt}=0=T^i\wedge\theta_{kti},
  \]
  given the symmetry properties of $\Theta_{ij}$ and $\theta_{ijk}$. Finally, we have the following result.
  \begin{lemma}
    Let $\nu:=\nu^ie_i$ be a $\mR^m$-valued $2$-form on $J^1\tau$, horizontal respect to $\tau_1$. Suppose further that\footnote{Otherwise, $\theta_{ijk}=0$.} $m>3$ and
    \[
    \nu^i\wedge\theta_{ijk}=0
    \]
    for all $j,k=1,\cdots,m$. Then $\nu=0$.
  \end{lemma}
  \begin{proof}
    Let us write $\nu$ in the following form
    \[
    \nu=\nu^{i}_{jk}\theta^j\wedge\theta^k\otimes e_i;
    \]
    then
    \begin{equation}\label{eq:ConditionCoordinatesNu}
      0=\nu^i\wedge\theta_{ikl}=2\left(\nu^{i}_{kl}\theta_i+\nu^{i}_{li}\theta_k-\nu^{i}_{ki}\theta_l\right).
    \end{equation}
  From this equation it results that
    \[
    \theta^k\wedge\nu^i\wedge\theta_{ikl}=2\left(\nu^{k}_{kl}+m\nu^{i}_{li}-\nu^{i}_{li}\right)\sigma_0=2\left(m-2\right)\nu^i_{li}\sigma_0,
    \]
    so, because $m>2$, we have that $\nu^i_{li}=0$, and plugging it back in Equation \eqref{eq:ConditionCoordinatesNu}, we obtain $\nu^i_{jk}=0$, as required.
  \end{proof}
  As a consequence, the equations of motion associated to these vertical vector fields will be
  \begin{equation}\label{eq:TorsionMomentaVanish}
    T^i=0=\Theta_{kl}.
  \end{equation}
As we promised, torsion constraint is recovered from a variational problem involving only the metricity condition. Moreover, the fact that the multimomenta $\Theta_{ij}$ vanish means in particular that the Lepage-equivalent problem constructed in Section \ref{sec:griff-vari-probl} for the Griffiths variational problem
  \[
  \left(J^1\tau,\lambda_{PG},\cI_{PG}\right)
  \]
  is necessarily contravariant.
\item Let us consider now contractions of the differential $d\lambda_\cL$ with vector fields of the form $A_{J^1\tau}+0$, for $A\in\kf$. We have
  \[
  \left(-1\right)^m\eta^{ik}\left(\eta^{rq}\eta_{li}\left.\Theta_{rj}\right|_{\beta}A^l_q-\left.\Theta_{ip}\right|_{\beta}A^p_j\right)\wedge\left(\omega_\pf\right)^j_k=0
  \]
  that does not give rise to additional restrictions. It was expected, because $\cL_{PG}$ and the forms that define the metricity condition are invariant for the $\kf$-action.
\item Finally, we have to consider contractions with elements $A_{J^1\tau}+0$ for $A\in\pf$; we obtain in this case
  \begin{multline*}
    0=\left[2\eta^{kp}A_k^i\theta_{il}-\left(\mathop{\text{tr}}{A}\right)\eta^{kp}\theta_{kl}\right]\wedge\Omega^l_p+\\
    +\left(-1\right)^{m-1}\eta^{ik}\left[d\Theta_{ij}+\eta^{rq}\eta_{li}\left.\Theta_{rj}\right|_{\beta}\wedge\left(\omega_\kf\right)^l_q-\left.\Theta_{ip}\right|_{\beta}\wedge\left(\omega_\kf\right)^p_j\right]A^j_k,
  \end{multline*}
  namely
  \begin{equation}\label{eq:EinsteinEqFirst}
    A_k^i\left(\eta^{kp}\theta_{il}-\frac{1}{2}\delta_i^k\eta^{qp}\theta_{ql}\right)\wedge\Omega^l_p=0
  \end{equation}
for every $A\in\pf$. Equivalently, it will have that
\[
\theta_{il}\wedge\Omega^l_k+\theta_{kl}\wedge\Omega^l_i-\eta_{ik}\left(\eta^{pq}\theta_{ql}\wedge\Omega^l_p\right)=0.
\]

Let us show that this system is equivalent to Einstein's equations in vacuum; for this, it is necessary to consider that, on a solution, we can write down
  \[
  \Omega^l_p=\Omega^l_{pab}\theta^a\wedge\theta^b
  \]
  for some functions $\Omega^l_{pab}$ such that $\Omega^l_{pab}+\Omega^l_{pba}=0$. Then we have that
  \[
  \left(\eta^{kp}\theta_{il}-\frac{1}{2}\delta_i^k\eta^{qp}\theta_{ql}\right)\wedge\Omega^l_p=-2\eta^{kp}\left(\Omega_{pi}-\frac{1}{2}\eta_{pi}\Omega\right)\sigma_0
  \]
  where the standard notation $\Omega_{ij}:=\Omega^l_{ilj},\Omega:=\eta^{ij}\Omega_{ij}$ was employed. Now, from Equation \eqref{eq:EinsteinEqFirst} and taking into account that $\Omega_k^l$ is the Riemann tensor of a Levi-Civita connection, it results that $\Omega_{ij}=\Omega_{ji}$, we must conclude that
  \[
  \Omega_{ij}-\frac{1}{2}\eta_{ij}\Omega=0,
  \]
  as required.
\end{itemize}

\subsection{Constraint algorithm}
\label{sec:constraint-algorithm}

Successful field theory formulations should provide not only the equations of motion, but also the set of conditions ensuring that these equations are involutive. The additional procedures intended to extract this set of conditions (``constraints'', as they are usually called) are unsurprisingly dubbed ``constraint algorithms'' \cite{de1996geometrical,zbMATH02233555}. It is the purpose, for example, of Proposition $2$ in \cite{Gaset:2017ahy} in the realm of the unified formalism for Einstein-Hilbert gravity. The constraint algorithm we will follow here is the one formulated in \cite{2013arXiv1309.4080C}, where it was referred to as \emph{Gotay algorithm} (see also \cite{Estabrook:2014hfa}). In short, this is essentially Cartan algorithm \cite{Hartley:1997:IAN:2274723.2275278,CartanBeginners} for the first prolongation of the EDS $\cJ$ shown below (see Equation \eqref{eq:EinsteinPalatiniEDS}); the set of constraints are thus obtained by annihilating the torsion of the sucessive prolongations. It means in particular that if the first prolongation of $\cJ$ is involutive, no further constraints will arise, and it will become involutive.

As we mentioned above, the equations of motion for the premultisymplectic system are represented by the exterior differential system
\begin{equation}\label{eq:EinsteinPalatiniEDS}
  \cJ:=\left<\Theta,\omega_\pf,T,\Omega^k_l\wedge\theta^l,\theta_{il}\wedge\Omega^l_k+\theta_{kl}\wedge\Omega^l_i-\eta_{ik}\left(\eta^{pq}\theta_{ql}\wedge\Omega^l_p\right)\right>.
\end{equation}
The form $\Omega^k_l\wedge\theta^l$ comes from the first Bianchi identity; also, from constraint $\omega_\pf=0$ it results that $\Omega$ takes values in $\kf$.

The involutivity of this system would imply that every $m$-plane $Z$ on $W_\cL$ annihilating $\cJ$ and such that $Z\lrcorner\sigma_0\not=0$ is the tangent space for some solution of the field equations, and so no further constraint would appear. We know that Einstein's equations, at least in dimension $4$, are involutive (see \cite{doi:10.1063/1.3305321,PhysRevD.71.044004} and references therein). Although it does not necessarily imply that the EDS $\cJ$ is involutive, it will imply that its first prolongation $\cJ^{\left(1\right)}$ does. In order to prove it, we need to recall Corollary \ref{cor:FirstBianchiConsequence}, that implies
\[
\Omega^l_i\wedge\theta_{lp}-\Omega^l_p\wedge\theta_{li}\equiv0\mod{T}
\]
for the $\kf$-valued form $\Omega$. Thus
\begin{equation}\label{eq:SameEquationsLagHam}
  \theta_{il}\wedge\Omega^l_k+\theta_{kl}\wedge\Omega^l_i-\eta_{ik}\left(\eta^{pq}\theta_{ql}\wedge\Omega^l_p\right)\equiv\frac{1}{2}\eta_{ip}\theta^p\wedge\eta^{qs}\theta_{kqr}\wedge\Omega^r_s\mod{T}.
\end{equation}
Now, the EDS considered in \cite{PhysRevD.71.044004} becomes in our notation
\[
\cI_{\text{PG}}=\left<\omega_\pf,T,\Omega^k_l\wedge\theta^l,\eta^{qk}\theta_{iql}\wedge\Omega^l_k\right>,
\]
meaning in particular that
\[
V_4\left(\cI_{\text{PG}}\right)=V_4\left(\cJ\right)
\]
as subsets of $G_4\left(TW_\cL,\nu\right)$. Therefore, if we use the constraint algorithm as it is presented in \cite{2013arXiv1309.4080C}, we have no additional constraints, and the constraint algorithm must stops.

\begin{remark}
  Equation\eqref{eq:SameEquationsLagHam} is another proof fo the fact that the equations of motion obtained in this article are equivalent to those in \cite{capriotti14:_differ_palat}.
\end{remark}

\section{Unified formalism for unimodular gravity}
\label{sec:mult-unim-grav}

It is interesting to note that the same scheme is useful when dealing with unimodular gravity \cite{doi:10.1119/1.1986321,PhysRevD.92.024036}. Recall that in this formulation, the space-time is endowed with a volume form, dubbed \emph{fundamental form}. In this case, and using the fact that in this formulation of relativity the fundamental volume form is conserved, we must consider a reduction of $LM$ to a subbundle $UM$ with structure group $SL\left(m\right)$ \cite{KN1}. This reduction consist into the elements of $LM$ which are constant when contracted with the fundamental form, and its structure group becomes $SL\left(m\right)$. Because $\mathfrak{sl}\left(m\right)$ consists into the traceless elements of $\mathfrak{gl}\left(m\right)$, the decomposition carried out in Section \ref{sec:cart-decomp-forms} induces a decomposition
\[
\mathfrak{sl}\left(m\right)=\kf\oplus\pf',
\]
where $\pf'$ is the set of traceless elements in $\pf$. Thus we have a embedding
\[
i_{UM}:UM\hookrightarrow LM
\]
that can be lifted to an embedding $j^1i_{UM}:J^1\tau_{UM}\rightarrow J^1\tau$, where $\tau_{UM}:UM\rightarrow M$ is the restriction of $\tau$ to $UM$. Using this map it is immediate to pullback the canonical form and the canonical torsion from $J^1\tau$ to $J^1\tau_{UM}$, and thus we can formulate a Griffiths variational problem for unimodular Palatini gravity. This variational problem should be compared with the first order variational problem for unimodular gravity as it is found in \cite{PhysRevD.92.024036}: In our case the degrees of freedom we are working with admit in principle arbitrary connections, not restricted to be Lorentz, and this restriction is implemented \emph{a posteriori} through the metricity condition.

Moreover, following the constructions above, we can find a unified formalism for unimodular gravity and deduce the same equations of motion when contracting the differential of the Lagrangian form for the Lepage-equivalent problem with the vertical vector fields $0+\delta\beta,\left(\theta^r,A_{UM}\right)^V+0$ and $A_{J^1\tau_{UM}}+0$ for $A\in\kf$. Now, besides the metricity condition, we have the \emph{equiaffinity condition} \cite{9780521441773}
\[
\omega^s_s=0
\]
on the connection.

The main changes regarding these equations of motion are twofold:
\begin{enumerate}
\item Equation \eqref{eq:EqTorsionTheta} is replaced by
  \[
  T^i\wedge\theta_{kli}+\Theta_{kl}=\mu\eta_{kl}
  \]
  for some unknown $m-1$-form $\mu$. From here it results again $T^i=0$ and
  \[
  \Theta_{kl}=\eta_{kl}\mu.
  \]
\item Although the form $\Theta$ is no longer zero, we will also end up with condition \eqref{eq:EinsteinEqFirst} for this new problem, because this condition concerns a \emph{traceless} symmetric matrix $A$ in unimodular case. So the equation of motion will involve a unknown function $\lambda$, due to the fact that the trace of the element
  \[
  \eta^{kp}\theta_{il}\wedge\Omega^l_p
  \]
  is no longer determined. Thus the associated equation of motion becomes
  \[
  \Omega_{ij}=\lambda\eta_{ij}.
  \]
  These equations are equivalent to the equation of unimodular gravity
  that we met in literature (see for example
  \cite{PhysRevD.92.024036}); in order to see this, take traces in
  both sides, so that
  \[
  \lambda=\frac{1}{4}\Omega
  \]
  and thus
  \[
  \Omega_{ij}-\frac{1}{4}\Omega\eta_{ij}=0.
  \]
\end{enumerate}

\appendix

\section{Forms on $J^1\tau$}
\label{sec:forms-j1tau}

We will establish in this appendix some basic properties regarding forms on the jet space $J^1\tau$ and some associated spaces. As a preliminary step, let us introduce the set of forms
\begin{align*}
\theta_{i_1\cdots i_p}&:=\frac{1}{\left(n-p\right)!}\epsilon_{i_1\cdots i_pi_{p+1}\cdots i_n}\theta^{i_{p+1}}\wedge\cdots\wedge\theta^{i_n}\cr
&=X_{i_p}\lrcorner\cdots\lrcorner X_{i_1}\lrcorner\sigma_0
\end{align*}
associated to the canonical form $\theta$ on $J^1\tau$; here $\sigma_0:=\theta^1\wedge\cdots\wedge\theta^m$. 

\subsection{Some identities involving canonical forms}\label{App:SomeIdent}

Let us use the algebraic properties of the vielbeins in order to settle the identities used in the present article. Recall that for $\alpha\in\Omega^p\left(X\right)$ we have that
\[
X\lrcorner\left(\alpha\wedge\beta\right)=\left(X\lrcorner\alpha\right)\wedge\beta+\left(-1\right)^p\alpha\wedge\left(X\lrcorner\beta\right).
\]
Then for
\[
\theta_{ij}=X_j\lrcorner X_i\lrcorner\sigma_0
\]
we obtain
\begin{align*}
  \theta^m\wedge\theta_{ij}&=\theta^m\wedge\left(X_j\lrcorner X_i\lrcorner\sigma_0\right)\\
  &=-X_j\lrcorner\left(\theta^m\wedge\left(X_i\lrcorner\sigma_0\right)\right)+\left(X_j\lrcorner\theta^m\right)\wedge\left(X_i\lrcorner\sigma_0\right)\\
  &=-X_j\lrcorner\left(-X_i\lrcorner\left(\theta^m\wedge\sigma_0\right)+\left(X_i\lrcorner\theta^m\right)\sigma_0\right)+\delta^m_j\theta_i\\
  &=\delta^m_j\theta_i-\delta^m_i\theta_j.
\end{align*}
Additionally
\begin{align*}
  \theta^m\wedge\theta_{ijk}&=\theta^m\wedge\left(X_k\lrcorner X_j\lrcorner X_i\lrcorner\sigma_0\right)\\
  &=-X_k\lrcorner\left(\theta^m\wedge\left(X_j\lrcorner X_i\lrcorner\sigma_0\right)\right)+\left(X_k\lrcorner\theta^m\right)\left(X_j\lrcorner X_i\lrcorner\sigma_0\right)\\
  &=-X_k\lrcorner\left(\delta^m_j\theta_i-\delta^m_i\theta_j\right)+\delta^m_k\theta_{ij}\\
  &=\delta^m_k\theta_{ij}-\delta^m_j\theta_{ik}+\delta^m_i\theta_{jk}
\end{align*}
and in general
\[
\theta^m\wedge\theta_{i_1\cdots i_p}=\sum_{k=1}^p\left(-1\right)^{k+1}\delta^m_{i_k}\theta_{i_1\cdots\widehat{i_k}\cdots i_p}.
\]

\subsection{Cartan decomposition and forms}
\label{sec:cart-decomp-forms}

The decomposition of $\mathfrak{gl}\left(m\right)$ associated to $\eta$ (see Section \ref{sec:geom-prel}) has useful properties. We will mention some of them in the present section of the Appendix.

\begin{lemma}\label{lem:SymAntiSymForms}
  $\gamma\in\kf$ (resp. $\gamma\in\pf$) if and only if $\gamma\eta:=\gamma_p^i\eta^{pj}E_{ij}$ (resp. $\eta\gamma:=\eta_{ip}\gamma_j^pE^{ij}$) takes values in the set of antisymmetric (resp. symmetric) matrices.
\end{lemma}

From this lemma we can deduce the following useful fact.

\begin{proposition}\label{prop:TraceWedge}
  Let $\gamma\in\Omega^p\left(N,\kf\right),\rho\in\Omega^q\left(N,\pf\right)$ be a pair of forms on $N$. Then
  \[
  \mathop{\text{Tr}}{\left(\gamma\wedge\rho\right)}=\gamma^k_p\wedge\rho^p_k=0.
  \]
\end{proposition}
\begin{proof}
  In fact,
  \begin{align*}
    \mathop{\text{Tr}}{\left(\gamma\wedge\rho\right)}&=\gamma^k_p\wedge\rho^p_k\\
    &=\eta_{kr}\gamma^r_p\wedge\eta^{ks}\rho_s^p,
  \end{align*}
  so we have that
  \[
  \mathop{\text{Tr}}{\left(\gamma\wedge\rho\right)}=\mu_{kl}\wedge\nu^{kl}
  \]
  where by Lemma \ref{lem:SymAntiSymForms}, $\mu:=\mu_{ij}E^{ij}$ takes values in the set of antisymmetric matrices and $\nu:=\nu^{ij}E_{ij}$ takes values in the set of symmetric matrices; so $\mathop{\text{Tr}}{\left(\gamma\wedge\rho\right)}=\mu_{kl}\wedge\nu^{kl}=0$, as required.
\end{proof}

\begin{proposition}\label{prop:DecompSquare}
  Let $\omega\in\Omega^n\left(N,\mathfrak{gl}\left(m\right)\right)$ be a $\mathfrak{gl}\left(m\right)$-valued $n$-form on $N$. Then
  \begin{align*}
    \left[\left(\omega\wedge\omega\right)_\kf\right]^i_j&=
             \begin{cases}
               \left(\omega_\pf\right)^i_p\wedge\left(\omega_\pf\right)^p_j+\left(\omega_\kf\right)^i_p\wedge\left(\omega_\kf\right)^p_j&n\text{ odd,}\cr
               \left(\omega_\pf\right)^i_p\wedge\left(\omega_\kf\right)^p_j+\left(\omega_\kf\right)^i_p\wedge\left(\omega_\pf\right)^p_j&n\text{ even,}
             \end{cases}\\
    \left[\left(\omega\wedge\omega\right)_\pf\right]^i_j&=
             \begin{cases}
               \left(\omega_\pf\right)^i_p\wedge\left(\omega_\kf\right)^p_j+\left(\omega_\kf\right)^i_p\wedge\left(\omega_\pf\right)^p_j&n\text{ odd,}\cr
               \left(\omega_\pf\right)^i_p\wedge\left(\omega_\pf\right)^p_j+\left(\omega_\kf\right)^i_p\wedge\left(\omega_\kf\right)^p_j&n\text{ even.}
             \end{cases}
  \end{align*}
\end{proposition}
\begin{proof}
  We have that
  \begin{align*}
    \eta_{jp}\eta^{iq}\omega^p_r\wedge\omega^r_q&=\eta_{jp}\omega^p_r\eta^{rk}\wedge\eta_{ks}\omega^s_q\eta^{iq}\\
    &=\left[\left(\omega_\pf\right)^k_j-\left(\omega_\kf\right)^k_j\right]\wedge\left[\left(\omega_\pf\right)^i_k-\left(\omega_\kf\right)^i_k\right]\\
    &=\left(-1\right)^n\left[\left(\omega_\kf\right)^i_k\wedge\left(\omega_\pf\right)^k_j+\left(\omega_\pf\right)^i_k\wedge\left(\omega_\kf\right)^k_j-\left(\omega_\pf\right)^i_k\wedge\left(\omega_\pf\right)^k_j-\left(\omega_\kf\right)^i_k\wedge\left(\omega_\kf\right)^k_j\right],
  \end{align*}
  and also
  \begin{align*}
  \omega^i_k\wedge\omega^k_j&=\left(\omega_\kf\right)^i_k\wedge\left(\omega_\pf\right)^k_j+\left(\omega_\pf\right)^i_k\wedge\left(\omega_\kf\right)^k_j+\left(\omega_\pf\right)^i_k\wedge\left(\omega_\pf\right)^k_j+\left(\omega_\kf\right)^i_k\wedge\left(\omega_\kf\right)^k_j.
  \end{align*}
  Therefore, we obtain
  \begin{align*}
    \left[\left(\omega\wedge\omega\right)_\kf\right]^i_j&=\frac{1}{2}\left(\omega^i_p\wedge\omega^p_j-\eta_{jp}\eta^{iq}\omega^p_r\wedge\omega^r_q\right)\\
    &=\left(\omega_\pf\right)^i_k\wedge\left(\omega_\pf\right)^k_j+\left(\omega_\kf\right)^i_k\wedge\left(\omega_\kf\right)^k_j
  \end{align*}
  for the $\kf$-projection, and
  \begin{align*}
    \left[\left(\omega\wedge\omega\right)_\pf\right]^i_j&=\frac{1}{2}\left(\omega^i_p\wedge\omega^p_j+\eta_{jp}\eta^{iq}\omega^p_r\wedge\omega^r_q\right)\\
    &=\left(\omega_\kf\right)^i_k\wedge\left(\omega_\pf\right)^k_j+\left(\omega_\pf\right)^i_k\wedge\left(\omega_\kf\right)^k_j,
  \end{align*}
  for the $\pf$-projection in the $n$ odd case, and
  \begin{align*}
    \left[\left(\omega\wedge\omega\right)_\kf\right]^i_j&=\frac{1}{2}\left(\omega^i_p\wedge\omega^p_j-\eta_{jp}\eta^{iq}\omega^p_r\wedge\omega^r_q\right)\\
    &=\left(\omega_\kf\right)^i_k\wedge\left(\omega_\pf\right)^k_j+\left(\omega_\pf\right)^i_k\wedge\left(\omega_\kf\right)^k_j\\
    \left[\left(\omega\wedge\omega\right)_\pf\right]^i_j&=\frac{1}{2}\left(\omega^i_p\wedge\omega^p_j+\eta_{jp}\eta^{iq}\omega^p_r\wedge\omega^r_q\right)\\
    &=\left(\omega_\pf\right)^i_k\wedge\left(\omega_\pf\right)^k_j+\left(\omega_\kf\right)^i_k\wedge\left(\omega_\kf\right)^k_j
  \end{align*}
  when $n$ is even, as required.
\end{proof}

\begin{remark}
  As we mention in the Remark \ref{rem:Signature-independence}, there is nothing special in the results of this section regarding the chosen signature. Everything could be proved with a more general signature $\left(p,q\right)$, for  $p+q=m$.
\end{remark}

\subsection{The canonical $k$-form on a bundle of forms}
\label{sec:canonical-k-form}

Let us study in this section the behaviour of a canonical form on a principal bundle respect to the lifted action. Namely, let $\pi:P\rightarrow N$ be a principal bundle with structure group $G$, $\left(V,\rho\right)$ a $G$-representation, and suppose further that there exists a pair of surjective submersions $q:P\rightarrow M,p:N\rightarrow M$, such that the following diagram is commutative
\begin{center}
  \begin{tikzcd}[column sep=1cm, row sep=1.3cm]
    P
    \arrow{rr}{\pi}
    \arrow[swap]{dr}{q}
    &
    &
    N
    \arrow{dl}{p}
    \\
    &
    M
    &
  \end{tikzcd}
\end{center}

Consider now the bundle
\[
\overline{\tau}^k_{n,q}:\wedge^k_{n,q}\left(T^*P\right)\otimes V\rightarrow P
\]
of $V$-valued $k$-forms on $P$ which are $n$-vertical respect to $q$; it means that $\alpha\in\wedge^k_{n,q}\left(J^1\tau\right)\otimes V$ if and only if
\[
\alpha\left(X_1,\cdots,X_k\right)=0
\]
whenever $n$ of the vectors $X_1,\cdots,X_k$ belong to $\mathop{\text{ker}}{Tq}$.

This bundle has a canonical $V$-valued $k$-form $\Theta^k_{n,q}$ defined through the formula
\[
\left.\Theta^k_{n,q}\right|_{\alpha}\left(Z_1,\cdots,Z_k\right):=\alpha\left(T_\alpha\overline{\tau}^k_{n,q}\left(Z_1\right),\cdots,T_\alpha\overline{\tau}^k_{n,q}\left(Z_k\right)\right).
\]
The bundle $\wedge^k_{n,q}\left(T^*P\right)\otimes V$ is a $G$-space, with action of an element $g\in G$ defined via
\[
\Phi^k_g\left(\alpha\right)\left(X_1,\cdots,X_k\right):=\rho\left(g\right)\cdot\left(\alpha\left(T_{u\cdot g}R_{g^{-1}}X_1,\cdots,T_{u\cdot g}R_{g^{-1}}X_k\right)\right),
\]
for $\alpha\in\wedge^k_{n,q}\left(T^*_uP\right)\otimes V$ and $X_1,\cdots,X_k\in T_{u\cdot g}P$. The canonical form has special proerties regarding this action.

\begin{lemma}
  The canonical $k$-form $\Theta^k_{n,q}$ is $G$-equivariant.
\end{lemma}
\begin{proof}
  It is necessary to recall the commutative diagram
  \begin{center}
    \begin{tikzcd}[column sep=1cm, row sep=1.3cm]
      T\left(\wedge^k_{n,q}\left(T^*P\right)\otimes V\right)
      \arrow{r}{T\Phi^k_g}
      \arrow{d}{T\overline{\tau}^k_{n,q}}
      &
      T\left(\wedge^k_{n,q}\left(T^*P\right)\otimes V\right)
      \arrow{d}{T\overline{\tau}^k_{n,q}}
      \\
      TP
      \arrow{r}{TR_{g}}
      &
      TP
    \end{tikzcd}
  \end{center}
  for every $g\in G$. Then
  \begin{align*}
    \left(\Phi_g^k\right)^*\left(\left.\Theta^k_{n,q}\right|_{\Phi^k_g\left(\alpha\right)}\right)&\left(Z_1,\cdots,Z_k\right)\\
       &=\left.\Theta^k_{n,q}\right|_{\Phi^k_g\left(\alpha\right)}\left(T_\alpha \Phi^k_gZ_1,\cdots,T_\alpha\Phi^k_gZ_k\right)\\
                                                                                                                 &=\Phi^k_g\left(\alpha\right)\left(\left(T_{\Phi^k_g\left(\alpha\right)}\overline{\tau}^k_{n,q}\circ T_\alpha \Phi^k_g\right)Z_1,\cdots,\left(T_{\Phi^k_g\left(\alpha\right)}\overline{\tau}^k_{n,q}\circ T_\alpha\Phi^k_gZ_k\right)\right)\\
                                                                                                                 &=\Phi^k_g\left(\alpha\right)\left(\left(T_{u} R_{g}\circ T_{\alpha}\overline{\tau}^k_{n,q}\right)Z_1,\cdots,\left(T_{u} R_{g}\circ T_{\alpha}\overline{\tau}^k_{n,q}\right)Z_k\right)\\
       &=\rho\left(g\right)\cdot\left(\alpha\left(T_{u\cdot g}R_{g^{-1}}\left(T_{u} R_{g}\circ T_{\alpha}\overline{\tau}^k_{n,q}\right)Z_1,\cdots,T_{u\cdot g}R_{g^{-1}}\left(T_{u} R_{g}\circ T_{\alpha}\overline{\tau}^k_{n,q}\right)Z_k\right)\right)\\
       &=\rho\left(g\right)\cdot\left(\alpha\left(T_{\alpha}\overline{\tau}^k_{n,q}Z_1,\cdots,T_{\alpha}\overline{\tau}^k_{n,q}Z_k\right)\right)\\
       &=\rho\left(g\right)\cdot\left(\left.\Theta^k_{n,q}\right|_{\alpha}\left(Z_1,\cdots,Z_k\right)\right),
  \end{align*}
  as required.
\end{proof}

\section{Local expressions on $J^1\tau$}\label{App:LocalExpressions}

Let us recall some local expressions regarding canonical coordinates on $J^1\tau$; we are quoting almost word-to-word the Appendix $B.3.4$ of \cite{capriotti14:_differ_palat}.

Let $U\subset M$ be a coordinate neighborhood and $\tau:LM\rightarrow M$ the canonical projection of the frame bundle; on $\tau^{-1}\left(U\right)$ can be defined the coordinate functions
\[
u\in \tau^{-1}\left(U\right)\mapsto\left(x^\mu\left(u\right),e^\nu_k\left(u\right)\right)
\]
where $x^\mu\equiv x^\mu\circ \tau$ and
\[
u=\left\{e_1^\mu\left.\frac{\partial}{\partial x^\mu}\right|_{\tau\left(u\right)},\cdots,e_n^\mu\left.\frac{\partial}{\partial x^\mu}\right|_{\tau\left(u\right)}\right\}.
\]
If $\bar{U}\subset M$ is another coordinate neighborhood such that $U\cap\bar{U}\not=\emptyset$ and $u\in U\cap\bar{U}$, then
\[
u=\left\{\bar{e}_1^\mu\left.\frac{\partial}{\partial \bar{x}^\mu}\right|_{\tau\left(u\right)},\cdots,\bar{e}_n^\mu\left.\frac{\partial}{\partial \bar{x}^\mu}\right|_{\tau\left(u\right)}\right\},
\]
and the coordinates change on $\tau^{-1}\left(U\right)\cap\tau^{-1}\left(\bar{U}\right)\subset LM$ can be given as
\begin{align*}
  \bar{x}^\mu&=\bar{x}^\mu\left(x^1,\cdots,x^n\right)\\
  \bar{e}_k^\mu&=\frac{\partial\bar{x}^\mu}{\partial x^\nu}e_k^\nu.
\end{align*}
On the jet space $J^1p$ of any bundle $p:E\rightarrow M$, the change of adapted coordinates given by the rule $\left(x^\mu,u^A\right)\mapsto\left(\bar{x}^\nu\left(x\right),\bar{u}^B\left(x,u\right)\right)$ on $E$, transform the induced coordinates on $J^1p$ accordingly to \cite{saunders89:_geomet_jet_bundl}\\
\[
\bar{u}^A_\mu=\left(\frac{\partial\bar{u}^A}{\partial u^B}u^B_\nu+\frac{\partial\bar{u}^A}{\partial{x}^\nu}\right)\frac{\partial x^\nu}{\partial\bar{x}^\mu}.
\]
By supposing that the induced coordinates on $J^1\tau$ are in the present case $\left(x^\mu,e^\mu_k,e^\mu_{k\nu}\right)$ and $\left(\bar{x}^\mu,\bar{e}^\mu_k,\bar{e}^\mu_{k\nu}\right)$, we will have that
\[
\bar{e}^\mu_{k\nu}=\left(\frac{\partial\bar{x}^\mu}{\partial x^\sigma}e^\sigma_{k\rho}+\frac{\partial^2\bar{x}^\mu}{\partial{x}^\rho\partial{x}^\sigma}{e}^\sigma_k\right)\frac{\partial x^\rho}{\partial\bar{x}^\nu}.
\]
Take note on the fact that the functions
\[
A_{\mu\nu}^\sigma:=-e^\sigma_{k\nu}e^k_\mu,
\]
where the quantities $e^k_\mu$ are uniquely determined by the conditions
\[
e_\mu^ke_k^\nu=\delta^\nu_\mu,
\]
transform accordingly to
\[
\bar{A}^\mu_{\rho\gamma}=-\frac{\partial\bar{x}^\mu}{\partial x^\nu}\frac{\partial x^\sigma}{\partial\bar{x}^\alpha}e^\alpha_{k\sigma}\bar{e}^k_\rho-\frac{\partial^2\bar{x}^\mu}{\partial{x}^\rho\partial{x}^\alpha}\frac{\partial x^\alpha}{\partial\bar{x}^\gamma}.
\]
But by using the previous definition, we can find the way in which $e^k_\mu$ and $\bar{e}^k_\mu$ are related, namely
\[
\bar{e}^k_\mu=\frac{\partial x^\gamma}{\partial \bar{x}^\mu}e_\gamma^k
\]
and therefore
\[
\bar{A}^\mu_{\delta\nu}=\frac{\partial\bar{x}^\mu}{\partial x^\sigma}\frac{\partial x^\rho}{\partial\bar{x}^\nu}\frac{\partial x^\gamma}{\partial\bar{x}^\delta}A^\sigma_{\gamma\rho}-\frac{\partial^2\bar{x}^\mu}{\partial{x}^\rho\partial{x}^\gamma}\frac{\partial x^\rho}{\partial \bar{x}^\nu}\frac{\partial x^\gamma}{\partial \bar{x}^\delta},
\]
which is the transformation rule for the Christoffel symbols, if the following identity
\[
\frac{\partial^2\bar{x}^\sigma}{\partial x^\rho\partial x^\gamma}\frac{\partial x^\rho}{\partial \bar{x}^\mu}\frac{\partial x^\gamma}{\partial \bar{x}^\nu}=-\frac{\partial^2 x^\rho}{\partial\bar{x}^\mu \partial\bar{x}^\nu}\frac{\partial \bar{x}^\sigma}{\partial x^\rho}.
\]
is used. So we are ready to calculate local expressions for the previously introduced canonical forms. First we have that
\[
\theta^k=e^k_\mu dx^\mu
\]
determines the components of the tautological form on $J^1\tau$, and the canonical connection form will result
\[
\omega^k_l=e^k_\mu\left(d e^\mu_l-e^\mu_{l\sigma}d x^\sigma\right).
\]
It is immediate to show that
\[
\bar{\theta}^k=\theta^k,
\]
and moreover
\begin{align*}
  \bar{\omega}^k_l&=\bar{e}^k_\mu\left(d \bar{e}^\mu_l-\bar{e}^\mu_{l\nu}d \bar{x}^\nu\right)\\
  &=\frac{\partial x^\gamma}{\partial\bar{x}^\mu}e^k_\gamma\left[d\left(\frac{\partial\bar{x}^\mu}{\partial x^\gamma}e^\gamma_l\right)-\left(\frac{\partial\bar{x}^\mu}{\partial x^\sigma}e^\sigma_{l\rho}+\frac{\partial^2\bar{x}^\mu}{\partial{x}^\rho\partial{x}^\sigma}{e}^\sigma_l\right)\frac{\partial x^\rho}{\partial\bar{x}^\nu}d\bar{x}^\nu\right]\\
  &=\frac{\partial x^\gamma}{\partial\bar{x}^\mu}e^k_\gamma\left(\frac{\partial\bar{x}^\mu}{\partial x^\gamma}d e^\gamma_l-\frac{\partial\bar{x}^\mu}{\partial x^\sigma}e^\sigma_{l\rho}d{x}^\rho\right)\\
  &=e^k_\gamma\left(d e^\gamma_l-e^\gamma_{l\rho}d{x}^\rho\right)\\
  &=\omega^k_l.
\end{align*}
The associated curvature form can be calculated according to the formula
\begin{align*}
  \Omega^k_l&:=d\omega^k_l+\omega^k_p\wedge\omega^p_l\\
  &=d\left[e^k_\gamma\left(d e^\gamma_l-e^\gamma_{l\rho}d{x}^\rho\right)\right]+e^k_\gamma\left(d e^\gamma_p-e^\gamma_{p\sigma}d{x}^\sigma\right)\wedge\left[e^p_\sigma\left(d e^\sigma_l-e^\sigma_{l\rho}d{x}^\rho\right)\right]\\
  &=d e^k_\gamma\wedge\left(d e^\gamma_l-e^\gamma_{l\rho}d{x}^\rho\right)-e^k_\gamma d e^\gamma_{l\rho}\wedge d{x}^\rho+\\
  &\qquad+e^k_\gamma e^p_\sigma\left[d e^\gamma_p\wedge d e^\sigma_l+\left(e^\gamma_{p\beta}d e^\sigma_l\wedge d x^\beta-e^\sigma_{l\beta}d e^\gamma_p\wedge d x^\beta\right)+e^\gamma_{p\beta}e^\sigma_{l\delta}d x^\beta\wedge d x^\delta\right]\\
  &=-e^\gamma_{l\rho}d e^k_\gamma\wedge d{x}^\rho-e^k_\gamma d e^\gamma_{l\rho}\wedge d{x}^\rho+\\
  &\qquad+e^k_\gamma e^p_\sigma\left[\left(e^\gamma_{p\beta}d e^\sigma_l\wedge d x^\beta-e^\sigma_{l\beta}d e^\gamma_p\wedge d x^\beta\right)+e^\gamma_{p\beta}e^\sigma_{l\delta}d x^\beta\wedge d x^\delta\right]
\end{align*}
where in the passage from the third to the fourth line it was used the identity
\[
d e^k_\gamma\wedge d e^\gamma_l+e^k_\gamma e^p_\sigma d e^\gamma_p\wedge d e^\sigma_l=0.
\]
Because of the identity
\[
e_\gamma^kd e^\gamma_p=-e^\gamma_pd e_\gamma^k
\]
we can reduce further the expression for $\Omega^k_l$
\[
\Omega^k_l =e^k_\gamma\left[-d e^\gamma_{l\rho}\wedge d{x}^\rho+e^p_\sigma\left(e^\gamma_{p\beta}d e^\sigma_l\wedge d x^\beta+e^\gamma_{p\beta}e^\sigma_{l\delta}d x^\beta\wedge d x^\delta\right)\right].
\]
Take note that
\begin{equation}\label{Eq:CurvaturaIntermedia}
e^l_\mu\Omega^k_l=e^k_\gamma\left(d A^\gamma_{\mu\rho}\wedge d{x}^\rho+A^\gamma_{\sigma\beta}A^\sigma_{\mu\delta}d x^\beta\wedge d x^\delta\right),
\end{equation}
so that if we fix a connection $\Gamma$ through its Christoffel symbols $\left(\Gamma^\mu_{\nu\sigma}\right)$ in the canonical basis $\left\{\partial/\partial x^\mu\right\}$, then we will have that $e^\gamma_k=\delta^\gamma_k$ and this formula reduces to
\[
\Omega^\mu_\nu:=e_k^\nu e^l_\mu\Omega^k_l=d\Gamma^\mu_{\nu\rho}\wedge d{x}^\rho+\Gamma^\mu_{\sigma\beta}\Gamma^\sigma_{\nu\delta}d x^\beta\wedge d x^\delta
\]
providing us with the usual formula for the connection in terms of the local coordinates.

Next we can provide a local expression for the map $\tilde\sigma_\Gamma:LM\rightarrow J^1\tau$. First we realize that a connection $\Gamma$ is locally described by a map
\[
\Gamma:x^\mu\mapsto\left(x^\mu,\Gamma^\sigma_{\mu\nu}\left(x\right)\right);
\]
in these terms, the map $\tilde\sigma_\Gamma$ is given by
\[
\tilde\sigma_\Gamma:\left(x^\mu,e^k_\nu\right)\mapsto\left(x^\mu,e^k_\nu,-e^\mu_k\Gamma_{\mu\nu}^\sigma\left(x\right)\right).
\]
It is convenient to stress about an abuse of language committed here: We are indicating with the same symbol $\tilde\sigma_\Gamma$ either the map itself and its local version. Nevertheless, we obtain the following local expression for the connection form associated to $\Gamma$, namely
\[
\left(\tilde\sigma_\Gamma^*\omega\right)_l^k=e^k_\mu\left(d e^\mu_l+e^\sigma_l\Gamma^\mu_{\sigma\rho}\left(x\right)d x^\rho\right).
\]
In our approach this equation is equivalent to the so called \emph{tetrad postulate}, which relates the components \emph{of the same connection} in the two representations provided by the theory developed here: As a section $\Gamma$ of the bundle of connections, and as an equivariant map $\tilde\sigma_\Gamma:LM\rightarrow J^1\tau$ such that the following diagram commutes 
\[
\begin{diagram}
  \node{LM}\arrow{e,t}{\tilde\sigma_\Gamma}\arrow{s,l}{\tau}\node{J^1\tau}\arrow{s,r}{p_{GL\left(m\right)}}\\
  \node{M}\arrow{e,b}{\Gamma}\node{C\left(LM\right)}
\end{diagram}
\]
According to the previous discussion, the pullback of these forms along the section $s:x^\mu\mapsto\left(x^\mu,e_k^\nu\left(x\right)\right)$ provides us with the expression for the connection forms associated to the underlying moving frame
\[
e_k\left(x\right):=e^\nu_k\left(x\right)\frac{\partial}{\partial x^k};
\]
in fact, given another such section $\bar{s}:x^\mu\mapsto\left(x^\mu,\bar{e}_k^\nu\left(x\right)\right)$, there exists a map $g:x^\mu\mapsto\left(g^k_l\left(x\right)\right)\in GL\left(m\right)$ relating them, namely
\[
\bar{e}^\mu_k\left(x\right)=g^l_k\left(x\right)e_l^\mu\left(x\right)
\]
and so
\[
\bar{s}^*\left(\tilde\sigma_\Gamma^*\omega\right)_l^k=h^k_pg^q_ls^*\left(\tilde\sigma_\Gamma^*\omega\right)_q^p+h^k_pd g^p_l.
\]
It allows us to answer the concerns raised in the introduction: The Palatini Lagrangian is a global form on $J^1\tau$, but this is false for its pullback along a local section. Namely, its global description needs the inclusion of information about the $1$-jet of the vielbein involved in the local representation of the connection.

\section{The $\mathfrak{gl}\left(m\right)\otimes\left(\mR^m\right)^*$-valued difference function $C$}
\label{sec:diff-tens-c_i}

In this section we will introduce a $\mathfrak{gl}\left(m\right)\otimes\left(\mR^m\right)^*$-valued function $C$ on $J^1\tau$, associated to a \emph{torsionless} connection $\sigma\in\Omega^1\left(LM,\mathfrak{gl}\left(m\right)\right)$. In fact, we define
\[
C\left(\xi\right):=\omega\left(\left(B\left(\xi\right)\right)^1\right)
\]
for every $\xi\in\mR^m$; in this formula $B\left(\xi\right)\in\mathfrak{X}\left(LM\right)$ is the standard horizontal vector field determined by $\xi$ and $\sigma$.

As we will see below, $C\left(j_x^1s\right)$ gives us the difference between the connection $\sigma$ evaluated at $u=s\left(x\right)$ and the connection at a point $u$ corresponding to $j_x^1s$. In our interpretation of Palatini gravity as a Griffiths variational problem on $J^1\tau$, this function corresponds to the variables determined by the well-known trick \cite[p. 44]{AshtekarNoPerturbative} of substracting a connection with zero torsion from the connection variables. 

We can give a coordinate version of these functions. In fact, because of the formula \eqref{eq:CanonicalConnection},
\[
\left.\omega\right|_{j_x^1s}=\left[T_{j_x^1s}\tau_{10}-T_xs\circ T_{j_x^1s}\tau_1\right]_{\mathfrak{gl}\left(m\right)},
\]
we have that
\[
\left.\omega\right|_{j_x^1s}\left(\left(B\left(e_i\right)\right)^1\right)=\left[B\left(e_i\right)-T_xs\left(s\left(x\right)\left(e_i\right)\right)\right]_{\mathfrak{gl}\left(m\right)}.
\]
Using the local expressions calculated in Section \ref{sec:local-expr-lift}, we have that
\begin{align*}
  T_xs\left(s\left(x\right)\left(e_i\right)\right)&=T_xs\left(e_i^\mu\frac{\partial}{\partial x^\mu}\right)\\
  &=e^\mu_i\left(\frac{\partial}{\partial x^\mu}+e^\nu_{j\mu}\frac{\partial}{\partial e^\nu_j}\right),
\end{align*}
and so
\[
\left.\omega\right|_{j_x^1s}\left(\left(B\left(e_i\right)\right)^1\right)=\left[e^\mu_i\left(e^\nu_{j\mu}-e^\sigma_j\Xi^\nu_{\mu\sigma}\right)\frac{\partial}{\partial e^\nu_j}\right]_{\mathfrak{gl}\left(m\right)}=\left[e^\mu_ie_\nu^k\left(e^\nu_{j\mu}-e^\sigma_j\Xi^\nu_{\mu\sigma}\right)\left(E_k^j\right)_{LM}\right]_{\mathfrak{gl}\left(m\right)};
\]
therefore
\[
C=e^\mu_ie_\nu^k\left(e^\nu_{j\mu}-e^\sigma_j\Xi^\nu_{\mu\sigma}\right)E_k^j\otimes e^i,
\]
for $\left\{e^1,\cdots,e^m\right\}\subset\left(\mR^m\right)^*$ the dual basis of $\left\{e_1,\cdots,e_m\right\}$.

\begin{lemma}\label{lem:CVerticalDuality}
  Let $C^k_{ij}$ be the coordinates functions of $C$ respect to the basis $\left\{E^k_i\otimes e^j\right\}$, where $\left\{e^j\right\}\subset\left(\mR^m\right)^*$ is the dual basis to $\left\{e_i\right\}$. Then
  \[
  \left(\theta^i,\left(E^j_l\right)_{LM}\right)^V\cdot C^p_{qr}=\delta^i_r\delta^j_q\delta_l^p.
  \]
\end{lemma}
\begin{proof}
  It follows using the local expressions of these objects.
\end{proof}

In order to formulate the following result, let us denote by $\rho$ the product representation of the adjoint and the transpose action of $GL\left(m\right)$ on $\mathfrak{gl}\left(m\right)\otimes\left(\mR^m\right)^*$. Given a representation $\left(V,\rho\right)$ of a Lie group $G$ on the vector space $V$, a $V$-valued function $f$ on a $G$-principal bundle $P$ is said to be \emph{of type $\rho$} if and only if
\[
f\left(u\cdot g\right)=\rho\left(g^{-1}\right)\cdot f\left(u\right)
\]
for every $u\in P$.

\begin{lemma}
  The map $C$ is a $\mathfrak{gl}\left(m\right)\otimes\left(\mR^m\right)^*$-valued function of type $\rho$.
\end{lemma}
\begin{proof}
  Let $\xi\in\mR^m$, $j_x^1s\in LM,u=s\left(x\right)\in LM$ and $g\in GL\left(m\right)$; recalling that the lift $j^1\Phi_t$ of the flow $\Phi_t:LM\rightarrow LM$ of the vector field $B\left(\xi\right)$ to $J^1\tau$ is the flow of $\left(B\left(\xi\right)\right)^1$, we have that
  \[
  \left(R_{g*}B\left(\xi\right)\right)^1=R_{g*}\left(B\left(\xi\right)\right)^1.
  \]
  Also, it must be remembered that
  \[
  R_{g*}\left(B\left(\xi\right)\right)=B\left(g^{-1}\cdot\xi\right).
  \]
  Therefore
  \begin{align*}
    C_{\left(j_x^1s\right)\cdot g}\left(\xi\right)&=\left.\omega\right|_{\left(j_x^1s\right)\cdot g}\left(\left.\left(B\left(\xi\right)\right)\right|^1_{\left(j_x^1s\right)\cdot g}\right)\\
    &=\left[\left(T_{\left(j_x^1s\right)\cdot g}R_{g}\right)^*\left.\omega\right|_{\left(j_x^1s\right)\cdot g}\right]\left(T_{\left(j_x^1s\right)\cdot g}R_{g^{-1}}\left(\left.\left(B\left(\xi\right)\right)^1\right|_{\left(j_x^1s\right)\cdot g}\right)\right)\\
    &=\left(\mathop{\text{Ad}_{g^{-1}}}{\left.\omega\right|_{j_x^1s}}\right)\left(\left(T_{u\cdot g}R_{g^{-1}}\left(\left.B\left(\xi\right)\right|_{u\cdot g}\right)\right)^1\right)\\
    &=\left(\mathop{\text{Ad}_{g^{-1}}}{\left.\omega\right|_{j_x^1s}}\right)\left(\left(\left.B\left(g\cdot\xi\right)\right|_{u}\right)^1\right),
  \end{align*}
  namely
  \[
  C_{\left(j_x^1s\right)\cdot g}\left(\xi\right)=\mathop{\text{Ad}_{g^{-1}}}{\left(C_{j_x^1s}\left(g\cdot\xi\right)\right)}
  \]
  for any $j_x^1s\in J^1\tau,\xi\in\mathfrak{gl}\left(m\right)$.
\end{proof}

\begin{remark}\label{rem:CAsASectionOfE}
  As pointed out in \cite[p. 76]{KN1}, the previous lemma means in particular that $C$ can be seen as a section of the bundle $E$ associated to the principal bundle $p^{J^1\tau}_{GL\left(m\right)}:J^1\tau\rightarrow C\left(LM\right)$ through the representation $\left(\mathfrak{gl}\left(m\right)\otimes\left(\mR^m\right)^*,\rho\right)$.
\end{remark}

Now, for any $A\in\mathfrak{gl}\left(m\right)$ and $j_x^1s\in J^1\tau$ we have that
\begin{align*}
  \left.A_{J^1\tau}\right|_{j_x^1s}\cdot C&=\left.\frac{d}{dt}\right|_{t=0}\left[C_{\left(j_x^1s\right)\cdot\left(\exp{\left(-tA\right)}\right)}\right]\\
  &=\left.\frac{d}{dt}\right|_{t=0}\left[\rho\left(\exp{\left(-tA\right)}\cdot C_{j_x^1s}\right)\right]\\
  &=-\left[A,C_{j_x^1s}\right]+A^*C_{j_x^1s}.
\end{align*}
We can use this formula in order to set the following result, which is consequence of Lemma 1 in \cite[p. 97]{KN1}. It will be relevant in the evaluation of the vector fields $\left(B\left(e_i\right)\right)^1$ on the curvature form $\Omega^j_k$.
\begin{corollary}\label{cor:VerticalDerivativesFunctonDifference}
  Let $Z\in\mathfrak{X}\left(J^1\tau\right)$ be an arbitrary vector field, and consider $j_x^1s\in J^1\tau$; further, let $v_\omega\left(Z\right)$ be the vertical part of $Z$ respect to the connection $\omega$ on the bundle $J^1\tau\rightarrow C\left(LM\right)$. Then
  \begin{enumerate}
  \item It results that
    \[
    \left.v_\omega\left(Z\right)\right|_{j_x^1s}\cdot C=-\left[\left.\omega\right|_{j_x^1s}\left(Z\right),C_{j_x^1s}\right]+\left(\left.\omega\right|_{j_x^1s}\left(Z\right)\right)^*C_{j_x^1s}.
    \]
  \item For any pair $Z_1,Z_2$ of horizontal vector fields for $\omega$ on $J^1\tau$, we have that
    \[
    \left.v_\omega\left(\left[Z_1,Z_2\right]\right)\right|_{j_x^1s}\cdot C=-2\left[\left.\Omega\right|_{j_x^1s}\left(Z_1,Z_2\right),C_{j_x^1s}\right]+2\left(\left.\Omega\right|_{j_x^1s}\left(Z_1,Z_2\right)\right)^*C_{j_x^1s}.
    \]
  \end{enumerate}
\end{corollary}

\section{The contraction of elements of $B$ with canonical forms on $W_\cL$}
\label{sec:contr-elem-b}

We will need how the elements of the basis $B$ contract with the forms $\Omega,T,\theta,\omega$ in the calculation of the equations of motion associated to the multisymplectic version of Palatini gravity. To this end, we devoted the following section.

\subsection{Contraction of the curvature $\Omega$ with elements of $B$}

Now, let us evaluate vectors $\left(B\left(e_i\right)\right)^1$ on the curvature $\Omega$; we have that
\[
\Omega^k_l\left(\left(B\left(e_i\right)\right)^1,\left(B\left(e_j\right)\right)^1\right)=d\omega^k_l\left(\left(B\left(e_i\right)\right)^1,\left(B\left(e_j\right)\right)^1\right)+C^k_{pi}C^p_{lj}-C^k_{pj}C^p_{li}
\]
where $C^i_{jk}$ are the components of $C$ in the basis $\left\{E^i_j\otimes e^k\right\}$. Now, the differential of the form $\omega^k_l$ reads
\[
d\omega^k_l\left(\left(B\left(e_i\right)\right)^1,\left(B\left(e_j\right)\right)^1\right)=
\left(B\left(e_i\right)\right)^1\cdot C^k_{lj}-\left(B\left(e_j\right)\right)^1\cdot C^k_{li}-\omega^k_l\left(\left[\left(B\left(e_i\right)\right)^1,\left(B\left(e_j\right)\right)^1\right]\right).
\]
Furthermore, using the identity
\[
\left[Y^1,Z^1\right]=\left(\left[Y,Z\right]\right)^1
\]
for any pair $Y,Z\in\mathfrak{X}\left(LM\right)$, we can rewrite the last term as follows
\[
\omega^k_l\left(\left[\left(B\left(e_i\right)\right)^1,\left(B\left(e_j\right)\right)^1\right]\right)=\omega^k_l\left(\left(\left[B\left(e_i\right),B\left(e_j\right)\right]\right)^1\right);
\]
recalling that $\sigma$ is torsionless, we have that $\left[B\left(e_i\right),B\left(e_j\right)\right]$ is vertical, and so \cite[Cor. (5.3)]{KN1}
\[
\left[B\left(e_i\right),B\left(e_j\right)\right]=-\left(R\left(B\left(e_i\right),B\left(e_j\right)\right)\right)_{LM},
\]
for $R$ the curvature form of $\sigma$, so that
\[
\omega^k_l\left(\left[\left(B\left(e_i\right)\right)^1,\left(B\left(e_j\right)\right)^1\right]\right)=-\omega^k_l\left(\left(R\left(B\left(e_i\right),B\left(e_j\right)\right)\right)_{J^1\tau}\right)=-R^k_l\left(B\left(e_i\right),B\left(e_j\right)\right)
\]
where the identity $\left(A_{LM}\right)^1=A_{J^1\tau}$, valid for any $A\in\mathfrak{gl}\left(m\right)$, was used. Finally
\[
\Omega^k_l\left(\left(B\left(e_i\right)\right)^1,\left(B\left(e_j\right)\right)^1\right)=\left(B\left(e_i\right)\right)^1\cdot C^k_{lj}-\left(B\left(e_j\right)\right)^1\cdot C^k_{li}+R^k_l\left(B\left(e_i\right),B\left(e_j\right)\right)+C^k_{pi}C^p_{lj}-C^k_{pj}C^p_{li}.
\]
This formula is the analogous of Equation $\left(2\right)$ in \cite[p. 44]{AshtekarNoPerturbative}.

Additionally, we have that
\[
\Omega^k_l\left(\left(B\left(e_i\right)\right)^1,\left(A_{LM}\right)^1\right)=d\omega^k_l\left(\left(B\left(e_i\right)\right)^1,\left(A_{LM}\right)^1\right)+C_{pi}^kA^p_l-A^k_pC^p_{li}
\]
for any $A\in\mathfrak{gl}\left(m\right)$. Again
\[
d\omega^k_l\left(\left(B\left(e_i\right)\right)^1,\left(A_{LM}\right)^1\right)=-\left(A_{LM}\right)^1\cdot C^k_{li}-\omega^k_l\left(\left[\left(B\left(e_i\right)\right)^1,\left(A_{LM}\right)^1\right]\right),
\]
and from
\[
\left[\left(B\left(e_i\right)\right)^1,\left(A_{LM}\right)^1\right]=\left(\left[B\left(e_i\right),A_{LM}\right]\right)^1=-\left(B\left(Ae_i\right)\right)^1=-A_i^p\left(B\left(e_p\right)\right)^1
\]
we can conclude with the formula
\[
d\omega^k_l\left(\left(B\left(e_i\right)\right)^1,\left(A_{LM}\right)^1\right)=-\left(A_{LM}\right)^1\cdot C^k_{li}+A^p_iC^k_{lp},
\]
namely
\[
\Omega^k_l\left(\left(B\left(e_i\right)\right)^1,\left(A_{LM}\right)^1\right)=-\left(A_{LM}\right)^1\cdot C^k_{li}+A^p_iC^k_{lp}+C_{pi}^kA^p_l-A^k_pC^p_{li}.
\]
But, from Equation \eqref{eq:InfGenOnC} we have that
\[
\left(A_{LM}\right)^1\cdot C^k_{li}=-A^p_lC^k_{pi}+C^p_{li}A_p^k+A^p_iC_{lp}^k
\]
and so
\[
\Omega^k_l\left(\left(B\left(e_i\right)\right)^1,\left(A_{LM}\right)^1\right)=0.
\]
It also can be proved in a more direct fashion by realizing that $\left(A_{LM}\right)^1=A_{J^1\tau}$ and $\Omega$ is a covariant derivative (and so, it annihilates on vertical vector fields for the projection $p_{GL\left(m\right)}^{J^1\tau}:J^1\tau\rightarrow J^1\tau/GL\left(m\right)=:C\left(LM\right)$).

Finally, we will compute the contraction of $\Omega^k_l$ with a pair consisting the vector field $\left(B\left(e_i\right)\right)^1$ and $\left(\theta^j,\left(A_{LM}\right)^V\right)$. Recalling that the forms $\theta,\omega$ are semibasic respect to the projection $\tau_{10}:J^1\tau\rightarrow LM$, we have that
\[
\Omega^k_l\left(\left(B\left(e_i\right)\right)^1,\left(\theta^j,A_{LM}\right)^V\right)=d\omega^k_l\left(\left(B\left(e_i\right)\right)^1,\left(\theta^j,A_{LM}\right)^V\right).
\]
Now
\[
d\omega^k_l\left(\left(B\left(e_i\right)\right)^1,\left(\theta^j,A_{LM}\right)^V\right)=-\left(\theta^j,A_{LM}\right)^V\cdot C^k_{li}-\omega^k_l\left(\left[\left(B\left(e_i\right)\right)^1,\left(\theta^j,A_{LM}\right)^V\right]\right),
\]
but from Lemma \ref{lem:LiftsBrackets}, we obtain that
\[
\left[\left(B\left(e_i\right)\right)^1,\left(\theta^j,A_{LM}\right)^V\right]\in\ker{T\tau_{10}},
\]
so for $A=E^p_q$,
\[
d\omega^k_l\left(\left(B\left(e_i\right)\right)^1,\left(\theta^j,A_{LM}\right)^V\right)=-\left(\theta^j,\left(E^p_q\right)_{LM}\right)^V\cdot C^k_{li}=-\delta^j_i\delta^p_l\delta^k_q;
\]
then we conclude that
\[
\Omega^k_l\left(\left(B\left(e_i\right)\right)^1,\left(\theta^j,\left(E^p_q\right)_{LM}\right)^V\right)=-\delta^j_i\delta^p_l\delta^k_q.
\]

Also, we will have that
\[
\Omega^q_p\left(\left(E^i_j\right)_{J^1\tau},\left(E^k_l\right)_{J^1\tau}\right)=d\omega^q_p\left(\left(E^i_j\right)_{J^1\tau},\left(E^k_l\right)_{J^1\tau}\right)
\]
because $\omega^q_r\left(\left(E^j_k\right)_{J^1\tau}\right)=\delta^q_k\delta^j_r$; additionally
\[
d\omega^q_p\left(\left(E^i_j\right)_{J^1\tau},\left(E^k_l\right)_{J^1\tau}\right)=-\omega^q_p\left(\left[\left(E^i_j\right)_{J^1\tau},\left(E^k_l\right)_{J^1\tau}\right]\right).
\]
Now from the fact that
\[
\left[E^i_j,E^k_l\right]=\delta^k_jE^i_l-\delta^i_lE^k_j,
\]
we obtain
\[
\left[\left(E^i_j\right)_{J^1\tau},\left(E^k_l\right)_{J^1\tau}\right]=-\left(\delta^k_jE^i_l-\delta^i_lE^k_j\right)_{J^1\tau}=\left(\delta^i_lE^k_j-\delta^k_jE^i_l\right)_{J^1\tau},
\]
namely
\[
\Omega^q_p\left(\left(E^i_j\right)_{J^1\tau},\left(E^k_l\right)_{J^1\tau}\right)=\left(\delta^k_jE^i_l-\delta^i_lE^k_j\right)^q_p=\delta^k_j\delta^i_p\delta^q_l-\delta^i_l\delta^k_p\delta^q_j.
\]

Another possible contraction is an infinitesimal generator with a vertical vector for $\tau_{10}$; it becomes
\[
\Omega^q_p\left(\left(E^i_j\right)_{J^1\tau},\left(\theta^r,\left(E^k_l\right)_{LM}\right)^V\right)=d\omega^q_p\left(\left(E^i_j\right)_{J^1\tau},\left(\theta^r,\left(E^k_l\right)_{LM}\right)^V\right).
\]
On the other hand,
\[
d\omega^q_p\left(\left(E^i_j\right)_{J^1\tau},\left(\theta^r,\left(E^k_l\right)_{LM}\right)^V\right)=-\omega^q_p\left(\left[\left(E^i_j\right)_{J^1\tau},\left(\theta^r,\left(E^k_l\right)_{LM}\right)^V\right]\right)=0
\]
because
\[
\left[\left(E^i_j\right)_{J^1\tau},\left(\theta^r,\left(E^k_l\right)_{LM}\right)^V\right]=\left(\theta^r,\left[\left(E^i_j\right)_{J^1\tau},\left(E^k_l\right)_{J^1\tau}\right]\right)^V
\]
as a consequence of Lemma \ref{lem:LiftsBrackets}; so
\[
\Omega^q_p\left(\left(E^i_j\right)_{J^1\tau},\left(\theta^r,\left(E^k_l\right)_{LM}\right)^V\right)=0.
\]

\subsection{Contraction of the torsion $T$ with elements of $B$}

Let us calculate the contraction of the elements of $B$ with the universal torsion $T$. We have that
\[
T^k\left(\left(B\left(e_i\right)\right)^1,\left(B\left(e_j\right)\right)^1\right)=d\theta^k\left(\left(B\left(e_i\right)\right)^1,\left(B\left(e_j\right)\right)^1\right)+C^k_{pi}\delta^p_j-C^k_{pj}\delta^p_i,
\]
because $\left(B\left(e_i\right)\right)^1\lrcorner\theta^k=\delta^k_i$. Additionally,
\[
d\theta^k\left(\left(B\left(e_i\right)\right)^1,\left(B\left(e_j\right)\right)^1\right)=
-\theta^k\left(\left[\left(B\left(e_i\right)\right)^1,\left(B\left(e_j\right)\right)^1\right]\right)=0
\]
using the fact that, because $\sigma$ is torsionless, the bracket
\[
\left[\left(B\left(e_i\right)\right)^1,\left(B\left(e_j\right)\right)^1\right]=\left(\left[B\left(e_i\right),B\left(e_j\right)\right]\right)^1=\left(R\left(B\left(e_i\right),B\left(e_j\right)\right)\right)_{J^1\tau}
\]
is a vector field tangent to the orbits of the action of $GL\left(m\right)$ on $J^1\tau$. Then
\[
T^k\left(\left(B\left(e_i\right)\right)^1,\left(B\left(e_j\right)\right)^1\right)=C^k_{ji}-C^k_{ij}.
\]

On the other hand, we can calculate
\begin{align*}
T^k\left(\left(B\left(e_i\right)\right)^1,\left(A_{LM}\right)^1\right)&=d\theta^k\left(\left(B\left(e_i\right)\right)^1,\left(A_{LM}\right)^1\right)-A^k_p\delta^p_i\\
  &=d\theta^k\left(\left(B\left(e_i\right)\right)^1,\left(A_{LM}\right)^1\right)-A^k_i.
\end{align*}
The differential becomes
\[
d\theta^k\left(\left(B\left(e_i\right)\right)^1,\left(A_{LM}\right)^1\right)=-\theta^k\left(\left[\left(B\left(e_i\right)\right)^1,\left(A_{LM}\right)^1\right]\right)
\]
where we used that $\theta^k\left(A_{J^1\tau}\right)=0$.

Therefore, from
\[
\left[\left(B\left(e_i\right)\right)^1,\left(A_{LM}\right)^1\right]=\left(\left[B\left(e_i\right),A_{LM}\right]\right)^1=-\left(B\left(Ae_i\right)\right)^1=-A_i^p\left(B\left(e_p\right)\right)^1
\]
we can conclude that
\[
d\theta^k\left(\left(B\left(e_i\right)\right)^1,\left(A_{LM}\right)^1\right)=A^k_i,
\]
and so
\[
T^k\left(\left(B\left(e_i\right)\right)^1,\left(A_{LM}\right)^1\right)=0.
\]
As before, it is only a check of the fact that $T^k$ is a exterior covariant derivative \cite{KN1}, and as such it must annihilates on infinitesimal generators of the $GL\left(m\right)$-action on $J^1\tau$.

Now let us calculate the contraction with elements tangent to the fibers of $\tau_{10}$. It results that
\[
T^k\left(\left(B\left(e_i\right)\right)^1,\left(\theta^j,A_{LM}\right)^V\right)=d\theta^k\left(\left(B\left(e_i\right)\right)^1,\left(\theta^j,A_{LM}\right)^V\right),
\]
as before. Moreover,
\[
d\theta^k\left(\left(B\left(e_i\right)\right)^1,\left(\theta^j,A_{LM}\right)^V\right)=-\theta^k\left(\left[\left(B\left(e_i\right)\right)^1,\left(\theta^j,A_{LM}\right)^V\right]\right),
\]
but from Lemma \ref{lem:LiftsBrackets}, we obtain that
\[
\left[\left(B\left(e_i\right)\right)^1,\left(\theta^j,A_{LM}\right)^V\right]\in\ker{T\tau_{10}},
\]
so
\[
d\theta^k\left(\left(B\left(e_i\right)\right)^1,\left(\theta^j,A_{LM}\right)^V\right)=0,
\]
and consequently
\[
T^k\left(\left(B\left(e_i\right)\right)^1,\left(\theta^j,A_{LM}\right)^V\right)=0.
\]

In a similar fashion can be proved that
\[
T^q\left(\left(E^i_j\right)_{J^1\tau},\left(E^k_l\right)_{J^1\tau}\right)=0=T^q\left(\left(E^i_j\right)_{J^1\tau},\left(\theta^r,\left(E^k_l\right)_{LM}\right)^V\right)
\]
and from here, that
\[
\left(\theta^r,\left(E^k_l\right)_{LM}\right)^V\lrcorner T^q=0.
\]

\bibliographystyle{alpha}
\newcommand{\etalchar}[1]{$^{#1}$}

\end{document}